\documentclass{article}
\usepackage[a4paper,left=4cm]{geometry}                  
\usepackage{lmodern}
\usepackage{hyperref}
\usepackage{tikz}
\usepackage{amsmath,amssymb,amsthm,bm}
\usepackage[english]{babel}

\usepackage{hyperref}
\hypersetup{
    colorlinks,
    citecolor=darkblue,
    filecolor=black,
    linkcolor=darkblue,
    urlcolor=black
}

\usepackage[capitalise]{cleveref}
\crefname{equation}{}{}

\makeatletter
\newtheorem*{rep@theorem}{\rep@title}
\newcommand{\newreptheorem}[2]{%
\newenvironment{rep#1}[1]{%
 \def\rep@title{#2 \ref{##1}}%
 \begin{rep@theorem}}%
 {\end{rep@theorem}}}
\makeatother

\newreptheorem{proposition}{Proposition}

\newtheorem{theorem}{Theorem}[section]
\newtheorem{lemma}[theorem]{Lemma}
\newtheorem{proposition}[theorem]{Proposition}
\newtheorem{corollary}[theorem]{Corollary}

\newcommand{\beq}{\begin{equation}}
\newcommand{\eeq}{\end{equation}}
\newcommand{\bea}{\begin{eqnarray}}
\newcommand{\eea}{\end{eqnarray}}
\newcommand{\im}{\textrm{im }}
\newcommand{\ket}[1]{| #1 \rangle}

\newcommand{\coh}{\mathcal{H}}

\title{\large \bf On the ground states of the $M_\ell$ models}
\date{\normalsize{\today}}
\author{\normalsize \sc{Liza Huijse}$^1$, and \sc{Christian Hagendorf}$^2$
\bigskip\\
{\normalsize
  \begin{minipage}{\textwidth}
  \begin{center}
  \textit{
  $^1$ Physics Department, SITP\\
382 Via Pueblo Mall,Varian Lab\\
Stanford University, Stanford CA, 94305-4060  
  \medskip\\
$^2$ Universit\'e Catholique de Louvain\\
  Institut de Recherche en Math\'ematique et Physique\\
  Chemin du Cyclotron 2, 1348 Louvain-la-Neuve, Belgium}
\bigskip\\
   \href{mailto:lhuijse@stanford.edu}{\texttt{lhuijse@stanford.edu}},
  \href{mailto:christian.hagendorf@uclouvain.be}{\normalsize \texttt{christian.hagendorf@uclouvain.be}}
    \end{center}
  \end{minipage}
}
} 

\definecolor{darkblue}{rgb}{0,0,.8}

\begin{document}
\maketitle

\begin{abstract}
We study the $M_\ell$ models for lattice fermions with supersymmetry introduced by Fendley, Nienhuis and Schoutens on one-dimensional chains. We determine the number of ground states as a function of the chain length as well as various boundary conditions by solving the corresponding cohomology problem. As an intermediate result we obtain the \textit{Cut (and Paste)} cohomology isomorphism, which maps the cohomology problem of chains, whose lengths differ by $\ell+2$ sites, onto each other. 
\end{abstract}

\section{Introduction}

The $M_\ell$ models for lattice fermions were introduced by Fendley, Nienhuis and Schoutens in \cite{fendley:03}. They describe fermions on graphs subjected to the \textit{exclusion constraint} which forbids connected particle clusters to contain more than $\ell$ particles. The interest in these models is that they allow an explicit realization of $\mathcal N=2$ supersymmetry on the lattice which leads to a Hamiltonian having all features of a typical quantum many-body problem such as particle hopping between nearby sites and local potential energies. The most studied case is $\ell=1$, corresponding to a model of fermions with nearest-neighbor exclusion. The original papers \cite{fendley:03_2,fendley:03} focus on one-dimensional homogeneous chains where the model is Bethe-ansatz solvable and it can be mapped to the XXZ spin chain with anisotropy $\Delta=-1/2$. Its relation to the continuum $\mathcal{N}=2$ superconformal field theory is well understood \cite{beccaria:05,huijse:11_2}. Spatial modulations in the interactions between the fermions lead to features that are close to the XYZ chain along the supersymmetric line \cite{bazhanov:05,mangazeev:10,bazhanov:06,blom:12}, and the ground states of the model are related to classically integrable equations \cite{fendley:10_1,beccaria:12}. In two dimensions the model exhibits a variety of physically interesting features such as charge frustration leading to extensive ground state entropy on some lattices \cite{fendley:05_2,JJ,engstrom,HvE,huijse:11_3} and a conjectured supertopological phase on others \cite{huijse:08_2, huijse:10_2}. The observed ground state degeneracies display interesting relations to combinatorics and cohomology theory. Indeed, Jonsson \cite{jonsson:06} computed the Witten index of the model on a torus by relating it to rhombus tilings. The tiling relation proved also to be a central tool in order to determine the exact number of ground states \cite{huijse:10}.

The models with $\ell \geq 2$ have so far only been considered on one-dimensional chains. The $M_2$ model on the chain is Bethe-ansatz solvable in the homogeneous case \cite{fendley:03} as well as in the inhomogeneous case as long as a hidden dynamical supersymmetry is preserved \cite{hagendorf:14}. The latter result can be generalized to the case of general $\ell$ as we will show in forthcoming work \cite{chlh:tbp}. The authors of \cite{fendley:03} argue that the $M_\ell$ models on periodic chains provide a lattice version of the $\ell$-th member of $\mathcal N=2$ superconformal minimal series. More recently, it was argued that away from the homogeneous case, the $M_\ell$ models with $\ell \geq 2$ support non-abelian topological excitations \cite{fokkema:15}. 

The purpose of this article is to determine the exact number of zero-energy ground states of the $M_\ell$ models on one-dimensional chains with various boundary conditions for general $\ell$. This is achieved through the analysis of the cohomology of the supercharges that define the model. In the next section, we define the model, introduce periodic and special boundary conditions, discuss the supercharges and their cohomology and review the relation between cohomology elements and zero-energy ground states. In \cref{sec:results}, we state the main results : the exact number of zero-energy states for periodic and special boundary conditions and an intriguing cohomology isomorphism that relates the cohomology of an $N$-site chain to that of a chain with $\ell+2$ more sites. The proofs of these results are presented in \cref{sec:proofs}. Some alternative proofs and many technical results are deferred to the appendices.

\section{The model}
\label{sec:model}
In this section, we recall the definition of the $M_\ell$ model \cite{fendley:03}. Moreover, we provide a short reminder of the relation to cohomology that we will use to determine the number of zero-energy states.

\subsection{Definition}
\paragraph{Hilbert space.} We consider spinless fermions on a one-dimensional chain of length $N$ with the exclusion rule that the number of consecutive occupied sites can at most be $\ell$. The chain can be open (a simple path graph) or closed (a cycle graph). Each site may either be empty or occupied by a fermion. We shall frequently depict particle configurations through their site occupation numbers: $0$ for an empty site, and $1$ for an occupied site. A typical configuration which is compatible with the exclusion rule for $\ell \geq 3$ is given by
\begin{equation*}
  101110011001
\end{equation*}
We call a sequence of $m$ consecutive occupied sites an $m$-cluster. Our example has two $1$-clusters, one $2$-cluster and one $3$-cluster if considered as a configuration of an open chain. For a closed chain the first and last site of the sequence are neighbors and hence there are two $2$-clusters and one $3$-cluster.

In addition to the usual periodic boundary conditions for closed chains and free boundary conditions for open chains, we will also consider special boundary conditions for open chains. Special boundary conditions, specified by $s = (c_1,c_N)$, are imposed by restricting the length of the connected particle cluster that contains the first site to be at most $c_1$ and the length of the connected particle cluster that contains the last site, site $N$, to be at most $c_N$, where $0 \leq c_1, c_N \leq \ell$. Note that $s=(\ell,\ell)$ corresponds to free boundary conditions. Furthermore, $s=(0,0)$ corresponds to free boundary conditions on a chain of length $N-2$. Incidentally, these special boundary conditions were independently considered in \cite{fokkema:15}.

For given $\ell$ (and $s = (c_1,c_N)$ for open chains) the allowed configurations label the basis vectors of the $M_\ell$ model's Hilbert space, which -- apart from the exclusion rules -- is a standard fermionic Fock space with canonical scalar product. We denote this Fock space for a chain of length $N$ with periodic $(p)$ or special $(s)$ boundary conditions by $V^{(p)}_{N}$ or $V^{(s)}_{N}$, respectively. When it is not important or clear from the context we will drop the label on $V^{(p/s)}_{N}$ that indicates the boundary conditions and sometimes also the label that indicates the chain length. The Fock space possesses a natural grading given by the fermion number. Let us denote by $V^{(p/s)}_{N,f}$ the subspace of $V^{(p/s)}_{N}$ with exactly $f$ fermions so that
\begin{equation*}
  V^{(p/s)}_{N} = \bigoplus_{f\geq 0} V^{(p/s)}_{N,f}.
\end{equation*}
We denote by $F$ the usual fermion number operator.

\paragraph{Supercharge and Hamiltonian.} The supercharge $Q$ is a nilpotent operator, i.e. $Q^2=0$, which inserts a fermion into the system so that $[F,Q]=Q$. Its adjoint $Q^\dagger$ removes a particle, i.e. $[F,Q^\dagger]=-Q^\dagger$. The action of $Q^{\dagger}$ on a simple basis vector is defined as follows: look at each site, and produce a new basis vector through removal of a particle (if possible), weight it by an amplitude $\lambda_{m,n}$ if it is the $n$-th member of an $m$-cluster, taking into account the usual fermionic string. Finally, take the sum over all contributions. The requirement $(Q^\dagger)^2=0$ (and thus $Q^2=0$) leads to the following constraints between the amplitudes \cite{fendley:03}:
\begin{equation}
  \lambda_{m,n}\lambda_{m-n,p-n}= \lambda_{m,p}\lambda_{p-1,n}, \quad 1\leq n < p\leq m.
   \label{eqn:compatibility}
\end{equation}
These difference equations can systematically be solved in terms of the parameters $\lambda_{m,1}$, provided that they are all non-zero. Indeed, defining $\mu_m = \prod_{j=1}^m \lambda_{j,1}$ for $m>0$ and $\mu_0=1$, we obtain
\begin{equation}
  \lambda_{m,n}=\frac{\mu_m}{\mu_{n-1}\mu_{m-n}}.
  \label{eqn:couplings}
\end{equation}
For simplicity, let us consider that the $\lambda_{m,1}$ are all real positive numbers. Given that one of them can be scaled to unity, the model depends thus on $\ell-1$ parameters. The Hamiltonian is defined as the anticommutator of the supercharge with its adjoint
\begin{equation}
  H=QQ^\dagger + Q^\dagger Q.
  \nonumber
\end{equation}
Its action on a given basis state leads to various hopping processes and particle swaps between neighboring particle clusters, as well as a potential energy. As opposed to the supercharges, the Hamiltonian conserves the fermion number $[H,F]=0$, and is therefore block-diagonal in the occupation number basis of $V$. Moreover it commutes with both the supercharge and its adjoint $[H,Q]=[H,Q^\dagger]=0$.

\subsection{Zero-energy states}
\paragraph{Relation to cohomology.} It follows from its definition that the Hamiltonian is a positive operator and therefore its eigenvalues are bounded from below by zero $E\geq 0$. Eigenstates with positive energy $E>0$ organize in doublets of the supersymmetry algebra $|\psi\rangle,Q|\psi\rangle$ with $Q^\dagger|\psi\rangle =0$ but $Q|\psi\rangle\neq 0$. The eigenvalue $E=0$ is however special. In fact, if $H$ possesses a zero-energy eigenstate $|\psi\rangle$ then it is automatically a ground state of the system. The requirement $H|\psi\rangle =0$ is equivalent to $Q|\psi\rangle = 0$ and $Q^\dagger|\psi\rangle=0$, i.e. the state forms a singlet. The main objective of this article is to determine the dimension of the space of such zero-energy states, which we refer to occasionally as the number of (linearly independent) ground states.

To achieve this we use the well-known fact that the zero-energy state space is isomorphic to the cohomology of $Q$. Let us introduce some terminology, which will be useful in the following, and then recall the correspondence. The Hilbert space of the model equipped with the supercharge defines quite naturally an ascending complex
\begin{equation*}
  V_{N,0}^{(p/s)} \overset{Q}{\longrightarrow} V_{N,1}^{(p/s)} \overset{Q}{\longrightarrow} \cdots \overset{Q}{\longrightarrow} V_{N,f-1}^{(p/s)} \overset{Q}{\longrightarrow} V_{N,f}^{(p/s)} \overset{Q}{\longrightarrow} V_{N,f+1}^{(p/s)} \overset{Q}{\longrightarrow} \cdots
\end{equation*}
Since $Q^2=0$, the image of $Q:V_{N,f-1}^{(p/s)}\to V_{N,f}^{(p/s)}$ (the $Q$-coboundaries) is a subspace of the kernel of $Q:V_{N,f}^{(p/s)}\to V_{N,f+1}^{(p/s)}$ (the $Q$-cocycles). The cohomology
\begin{equation*}
  \coh_Q^f(V_{N}) = \ker \{Q:V_{N,f}^{(p/s)}\to V_{N,f+1}^{(p/s)}\}/\im \{Q:V_{N,f-1}^{(p/s)}\to V_{N,f}^{(p/s)}\}
\end{equation*}
measures to which extent there are non-trivial elements in the kernel by taking the quotient by the image. We shall often just write $\coh^f_Q$ when it is clear that $Q$ acts on the complex defined through $V_{N}$. The cohomology of the complex is given by
\begin{equation*}
  \coh_Q = \bigoplus_{f\geq 0} \coh_Q^f.
\end{equation*}
By definition its elements can be represented by states $|\psi\rangle \in V_N^{(p/s)}$ with $Q|\psi\rangle=0$, so-called representatives, up to coboundaries $Q|\phi\rangle$ where $|\phi\rangle \in V_N^{(p/s)}$. We denote by $[|\psi\rangle]_{\coh_Q}\in \coh_Q$ the corresponding equivalence class in the cohomology. The trivial element  $[0]_{\coh_Q} \in \mathcal H_Q$ (zero) is thus represented by a coboundary $Q|\phi\rangle$.

Two linearly independent ground states of the Hamiltonian, $H$, correspond to two linearly independent elements of $\coh_Q$, and vice-versa \cite{witten:82}. Hence, without explicit diagonalization of $H$ one may determine the number of its zero-energy ground states by computing the dimension of $\coh_Q$. Intuitively, the correspondence between the zero-energy states and the cohomology can be understood as follows. Organize the Hilbert space, $V$, in terms of eigenstates of the Hamiltonian. Consider the kernel of $Q$. It contains all zero-energy states. Moreover, given a doublet $|\psi\rangle,Q|\psi\rangle$ only $Q|\psi\rangle$ is in the kernel. However, this state is also in the image of $Q$. Taking the quotient by this image amounts to setting them to zero and one is left only with the zero-energy states.

\paragraph{Independence of the parameters of the model.} We are thus interested in the dimension of $\coh_Q$. It is legitimate to ask whether it can change when the parameters $\lambda_{m,1},\,m=1,\dots,\ell,$ are varied. Let us show that as long as they are all non-zero this is not the case. To this end, consider an invertible transformation, $M$, on the Hilbert space $V$, and define a new supercharge through conjugation $Q'=MQM^{-1}$. It is known that $\dim \coh_{Q'}=\dim \coh_Q$ \cite{witten:82,masson:08}. In the present situation, the simple structure of the parameters \eqref{eqn:couplings} allows to find a diagonal transformation $M$ which trivializes the couplings. Suppose that $|\alpha\rangle$ is a simple basis state with $m_k$ clusters of length $k=1,\dots,\ell$. If we define $M|\alpha\rangle = \prod_{k=1}^\ell \mu_k^{-m_k}|\alpha\rangle$ then $Q'$ is a supercharge of the same type as $Q$, but with coupling constants $\lambda'_{m,n}=1$ for all $m,n$. $M$ is invertible as long as none of the $\mu_k$ vanishes. We conclude therefore that the number of ground states is independent of the choice of the model parameters provided that they are non-zero. If some of them vanish however, the number of ground states may be different: an example is the case where they all vanish. Then the Hamiltonian is identically zero and annihilates trivially the full Hilbert space.

\paragraph{General strategy and the tic-tac-toe lemma.} To determine the dimension of $\mathcal H_Q$, we follow the strategy of previous works and compute an approximation to $\coh_Q$ in a finite number of steps. Let us explain one such step. First, we introduce a bit of notation. Let $S=\{1,\dots,N\}$ be the collection of all sites of the chain. We divide it into two disjoint sets $S=S_1 \sqcup S_2$. The precise choice of the subdivision depends on the boundary conditions as we shall see below. Let us consider the restrictions of the supercharge to the two sublattices
\begin{equation*}
  Q_1= \left. Q\right|_{S_1}, \quad Q_2= \left. Q\right|_{S_2}.
\end{equation*}
The original supercharge is given as the sum $Q=Q_1+Q_2$. Likewise, the total fermion number $f=f_1+f_2$ is the sum of the fermion numbers on the sublattices $f_1,f_2$. Finally, one verifies that the two sublattice supercharges obey the anticommutation relation
\begin{equation*}
  Q_1Q_2+Q_2Q_1=0.
\end{equation*} 
The subdivision leads therefore naturally to a so-called double-complex illustrated in \cref{fig:doublecomplex}. Each row and column in this figure constitutes its own ascending complex. It is legitimate to ask whether $\coh_Q$ can be obtained from the cohomologies of these complexes. The answer is in general non-trivial. For our purposes, a particular result which is known as the \textit{tic-tac-toe lemma} will be sufficient:
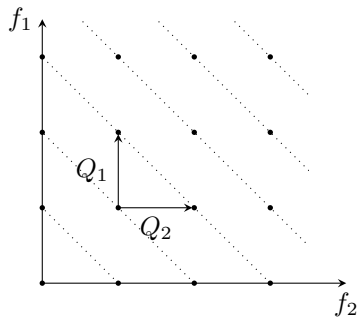
\begin{figure}[h]
  \centering
  \begin{tikzpicture}[>=stealth]
    \foreach \x in {0,1,...,3}
    { \foreach \y in {0,1,...,3}
        \fill (\x cm, \y cm) circle (1pt);
    }
    \draw[<->] (4,0) -- (0,0) -- (0,3.5);
    \draw[<->] (1,1.975) -- (1,1) -- (1.975,1);
    \draw (4,0) node[below] {$f_2$};
    \draw (0,3.5) node[left] {$f_1$};
    \draw (1,1.5) node [left] {$Q_1$};
    \draw (1.5,1) node [below] {$Q_2$};
    
    \begin{scope}
    \clip (0,3.5) rectangle (3.5,0);
    \foreach \x in {1,2,...,6}
    { 
      \draw[dotted] (\x,0) -- (0,\x);
    }
    \end{scope}

  \end{tikzpicture}
  \caption{Illustration of the double complex and action of the supercharges $Q_1$ and $Q_2$ associated to the sublattices $S_1$ and $S_2$. Each dot corresponds to the subspace of the model's Hilbert space with $f_1$ particles on $S_1$ and $f_2$ particles on $S_2$. The dotted diagonal lines correspond to constant total fermion number $f=f_1 +f_2$. }
  \label{fig:doublecomplex}
\end{figure}
\begin{lemma}[Tic-tac-toe]\label{prop:ttt}
  Given the double complex defined above, if $\coh_{21}\equiv \coh_{Q_2}(\coh_{Q_1})$ has entries in only one row, then $\coh_{Q_2}(\coh_{Q_1})$ is isomorphic to $\coh_Q$.
\end{lemma}
Here $\coh_{Q_2}(\coh_{Q_1}) \equiv \coh_{Q_2}(\coh_{Q_1}(V))$ is the cohomology of $Q_2$ acting on $\coh_{Q_1}$. For the proof and the construction of the isomorphism, we refer the reader to \cite{BottTu82}. Here, we will use that the isomorphism allows to find the number of ground states by computing the dimension of the two-step cohomology, $\coh_{Q_2}(\coh_{Q_1})$, which is often simpler to handle than the cohomology of $Q$. From the definition of $\coh_{Q_2}(\coh_{Q_1})$ it follows that is elements can be represented by states $|\psi\rangle \in V_N^{(p/s)}$ which obey
\begin{subequations}
\begin{align}
& Q_1 \ket{\psi}  = 0, \label{eqn:CondRepH21First} \\
& Q_2 \ket{\psi} = Q_1 \ket{\phi}.
\label{eqn:CondRepH21Second}
\end{align}
\label{eqn:CondRepH21}%
\end{subequations}%
The first statement simply reflects that all elements in $\coh_{Q_2}(\coh_{Q_1})$ are in $\coh_{Q_1}$ and thus in the kernel of $Q_1$. The second statement reflects that every element is in the kernel of $Q_2$ within $\coh_{Q_1}$, that is, $Q_2$ maps it to the trivial equivalence class of $\coh_{Q_1}$. The trivial equivalence class $[0]_{\coh_{21}}$ can be represented by states of the form
\begin{equation}
  |\psi\rangle = Q_2|\alpha\rangle + Q_1|\beta\rangle, \quad \text{with} \quad Q_1|\alpha\rangle = 0.
  \label{eqn:ZeroInH21}
\end{equation}

\section{Results}\label{sec:results}

In this work we obtain results on the zero-energy ground states of the $M_{\ell}$ models with various boundary conditions by computing the corresponding cohomology through an application of the tic-tac-toe lemma. For the $M_1$ model \cite{fendley:03_2} and more recently for the $M_2$ model \cite{hagendorf:14} the cohomology problem was analyzed for periodic closed and free open boundary conditions. Here we extend these results to general $\ell$. Furthermore, we include special boundary conditions, $s=(c_1,c_N)$, for open chains, where -- in addition to the overall constraint on the cluster length -- the lengths of the clusters that include the first or last site are restricted. We present the main results here, the proofs are postponed to the next sections and the appendix. 

Our results rely on the existence of a powerful cohomology isomorphism, which is an interesting result in its own right :

\begin{theorem}[Cut (and paste)]
\label{thm:cutandpaste}
The cohomology, $\coh^{f}_{Q}(V^{(p/s)}_{N})$, of a chain of length $N>2\ell+2$ with periodic boundary conditions (p) or special boundary conditions, $s=(c_1,c_N)$, at grade $f \geq \ell $ is isomorphic to the cohomology, $\coh^{f-\ell}_{Q}(V^{(p/s)}_{N'})$, of a chain of length $N' \equiv N-(\ell+2)$ with periodic boundary conditions (p) or special boundary conditions $s=(c_1,c_N' = c_N)$, respectively, at grade $f-\ell$.
\end{theorem}

We shall prove this theorem through the explicit construction of an isomorphism between 
the cohomology for an $N$ site chain with boundary conditions $(p)$ or $(s)$ and that of an $N-(\ell+2)$ site chain with the same boundary conditions. The reason we call it \textit{cut (and paste)} is that the isomorphism is realized by cutting out the first $(\ell+2)$ sites and only for closed chains with periodic boundary conditions the ends are pasted together again. It is not obvious that this isomorphism exists : the propagation of the boundary conditions is quite remarkable.
It makes the periodic structure in the number of zero-energy ground states as a function of chain length very explicit. It would be interesting to investigate if this isomorphism can be used to reveal possible self-similarity properties of the ground states as the chain length is increased by steps of $\ell+2$, or even to prove the scale-free properties of ground-state correlation functions observed in \cite{fendley:10,fendley:10_1,beccaria:12}.We leave these questions to future studies. Here, our aim is to apply \cref{thm:cutandpaste} in order to find the number of zero-energy states of the $M_\ell$ model. For closed chains we obtain the following result :

\begin{theorem}
\label{prop:closedchain}
The Hamiltonian of the closed periodic chain with $N=n(\ell+2)+p+1$ sites, $n \geq 0$, has $1$ zero-energy ground state with fermion number $f=n\ell+p$ for $p=0,\dots,\ell$ and has $\ell+1$ zero-energy ground states with fermion number $f=(n+1)\ell$ for $p=\ell+1$.
\end{theorem}

One easily verifies that this result is compatible with the results obtained for the Witten index in \cite{fendley:03} and with the zero-energy ground state results for $\ell=1,2$ \cite{fendley:03_2,hagendorf:14}. For open chains with special boundary conditions we obtain the following.
\begin{theorem}\label{prop:specialchain}
The Hamiltonian of the open chain with special boundary conditions given by $s=(c_1,c_N)$ and length $N=n(\ell+2)+p+1$ with $n \geq 1$ and $0 \leq p \leq \ell+1$, has one zero-energy ground state for (i) $c_1,c_N \geq p$ and $c_1+c_N \leq \ell+p$ with fermion number $f=n\ell+p$ when $0 \leq p \leq \ell$ and one for (ii) $c_1,c_N < p$ and $c_1+c_N \geq p-1$  with fermion number $f=n\ell+p-1$ when $1 \leq p \leq \ell+1$. It has no zero-energy state otherwise. 
\end{theorem}
\begin{figure}[ht]
\begin{center}
\includegraphics[width = 0.5\textwidth]{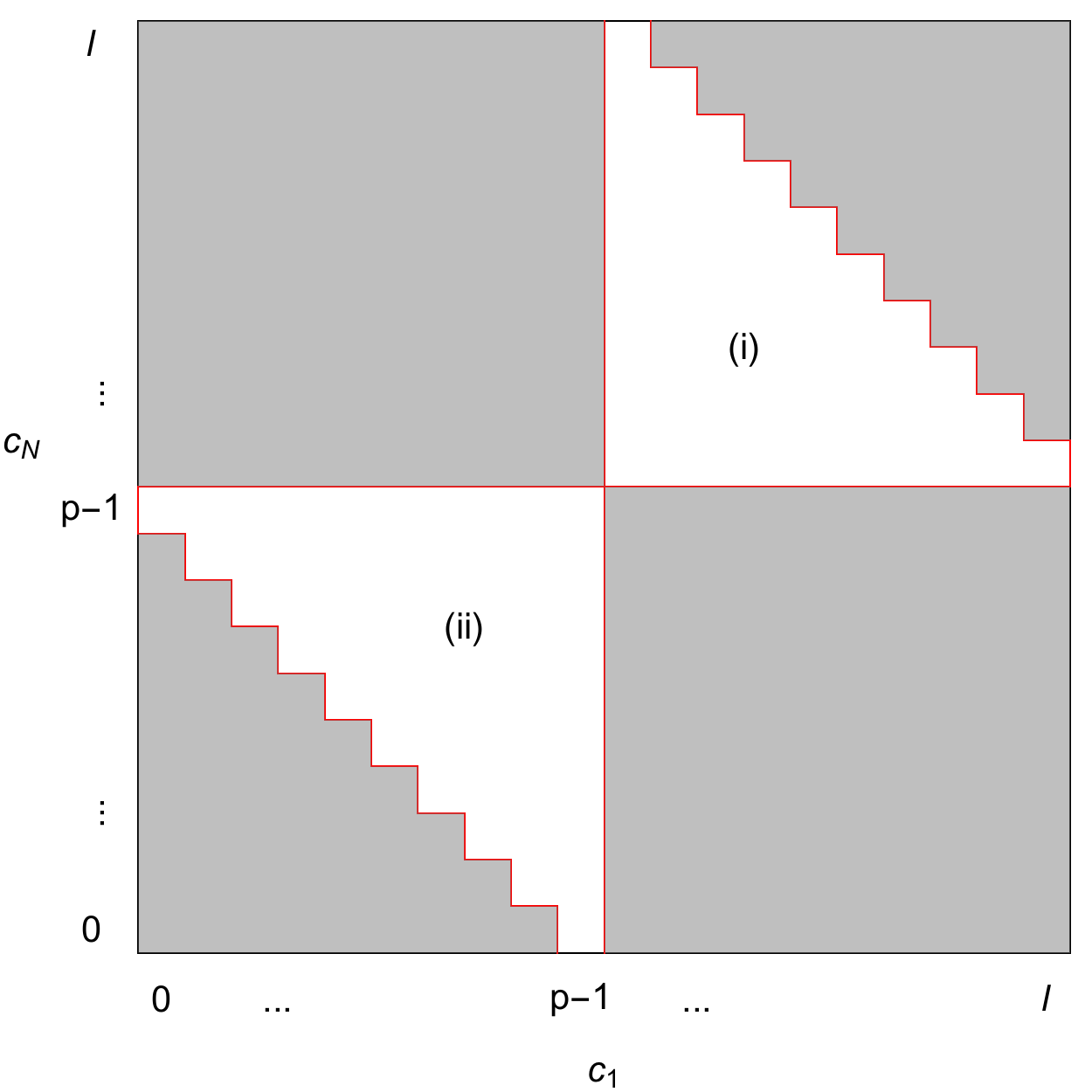}
\caption{We illustrate for given $p$ the special boundary conditions for which there are zero-energy ground states. The figure denotes a 2d grid with each square labeled by $(c_1,c_N)$ and $0<c_1,c_N \leq \ell$. There is a unique zero-energy state for the boundary conditions that satisfy (i) $c_1,c_N \geq p$ and $c_1+c_N \leq \ell+p$ and (ii) $c_1,c_N < p$ and $c_1+c_N \geq p-1$ corresponding to the white regions indicated by (i) and (ii), respectively. The regions corresponding to the boundary conditions for which there are no zero-energy ground states are gray shaded. \label{fig:special}}
\end{center}
\end{figure}
We illustrate this result pictorially in \cref{fig:special}. Note that the result for the open chain with free boundary conditions is included in the proposition as it corresponds to special boundary conditions with $s=(\ell,\ell)$. For this particular case there are zero-energy ground states only if $p=\ell$ or $p=\ell+1$. This is again consistent with previous results for $\ell=1,2$ \cite{fendley:03_2,hagendorf:14}. Furthermore, the special boundary conditions for $\ell=2$ were also considered in \cite{fokkema:15}, since they play a role in selecting out certain topological sectors. Again the results are consistent. Finally, the results for $n=0$ can be found in \cref{lem:special} below.

The proofs of these theorems are obtained in several steps, invoking short chain results and a few intermediate results. However, for two special cases we found short and simple alternative proofs. These simpler proofs, which exist for closed periodic chains with length $N=n(\ell+2)$ and for open chains with free boundary conditions and general length $N$, can be found in \cref{app:altproofs}. For the general case, the \textit{cut (and paste)} isomorphism gives us the dimension of the space of zero-energy states for chains of length $N>2\ell+2$ provided that we have the solutions for $1 \leq N \leq 2\ell+2$. The computation of these short chain cohomologies is quite technical
and is therefore deferred to \cref{app:short}.

\section{Proofs}\label{sec:proofs}

In this section we prove the theorems stated in the previous section. In particular, we construct explicitly the \textit{cut (and paste)} cohomology isomorphism for the $M_{\ell}$ models,  \cref{thm:cutandpaste}, maps the cohomology of $Q$ on a chain of length $N$ to the cohomology of $Q$ on a chain of length $N'=N-\ell-2$. This isomorphism is essential in the proofs of the statements on the zero-energy ground states of the $M_{\ell}$ models given in \cref{prop:closedchain,prop:specialchain}, since it allows us to obtain the cohomology of an $N$-site open or closed chain for any  $N>2\ell+2$ provided that we have the solution for chains of lengths $1\leq N \leq 2(\ell+1)$. These short chain results are given in \cref{sec:shortlems}. This is followed by the proof of \cref{thm:cutandpaste} in \cref{sec:iso}. Finally, the proofs of \cref{prop:closedchain,prop:specialchain} are given in sections \cref{sec:closedchain,sec:specialchain}, respectively. 

\subsection{Short chains}
\label{sec:shortlems}
The propositions stated in this section serve in the proofs of the \textit{cut (and paste)} cohomology isomorphism and our results on ground states of the $M_\ell$ models. Their proofs can be found in \cref{app:short}.

\begin{proposition}\label{lem:special}
The cohomology, $\coh_Q (V^{(s)}_N)$, of a chain of length $1\leq N\leq \ell+2$ with special boundary conditions given by $s=(c_1,c_N)$ with $0\leq c_1,c_N \leq \ell$ has dimension one
when (i) $c_1=N-1$ and $c_N > N-1$ or $c_N=N-1$ and $c_1 \geq N-1$ for $1\leq N\leq \ell+1$ and when (ii) $c_1,c_N \leq N-2$ and $c_1+c_N \geq N-2$ for $2 \leq N \leq \ell+2$.  Otherwise, its dimension is zero. The number of particles of its non-trivial elements is (i) $N-1$ and (ii) $N-2$, respectively, and the non-trivial elements can be represented by
\begin{equation*}
  \text{(i)}\quad \ket{01\dots1} \quad \text{and} \quad \text{(ii)}\quad \ket{01\dots10}. 
\end{equation*}
\end{proposition}

\begin{proposition}\label{lem:specialshort2}
The cohomology, $\coh_Q (V^{(s)}_N)$, of a chain of length $\ell+ 3 \leq N\leq 2\ell+2$ with special boundary conditions given by $s=(c_1,c_N)$ with $0\leq c_1,c_N \leq \ell$ has dimension one when (i) $c_1,c_N \geq N-\ell-3$ and $c_1+c_N < N-2$ and when (ii) $c_1,c_N < N-\ell-3$ and $c_1+c_N \geq N-\ell-4$ for $\ell+3 < N \leq 2\ell+2$. Otherwise, its dimension is zero. The number of particles of the non-trivial elements is (i) $N-3$ and (ii) $N-4$, respectively.
\end{proposition}

\begin{proposition}\label{lem:periodicshort1}
The cohomology, $\coh_Q (V^{(p)}_N)$, of a chain with periodic boundary conditions has dimension one for $1\leq N\leq \ell+1$ and dimension $\ell+1$ for $N=\ell+2$. The non-trivial elements contain $N-1$ particles.
\end{proposition}

\begin{proposition}\label{lem:periodicshort2}
The cohomology, $\coh_Q (V^{(p)}_N)$, of a chain of length $\ell+3\leq N\leq 2\ell+2$ with periodic boundary conditions has dimension one. The non-trivial elements contain $N-3$ particles.
\end{proposition}

\subsection{Cut (and paste)}
\label{sec:iso}
In this section we present the proof of the \textit{cut (and paste)} cohomology isomorphism of \cref{thm:cutandpaste}.
It is the result of a number of steps whose basic idea is to divide the chain into two parts, and relate the cohomology of the full chain to those of the subchains by using the tic-tac-toe lemma. In the course of this procedure, we prove the existence of a recursive structure of the cohomology groups.

\paragraph{Notation.} 
Let us first introduce some notation. In this sections we fix integers $N,\tilde N$ such that
\begin{equation*}
  N > 2(\ell+1), \quad \text{and}\quad \tilde N = N-(\ell+2) > \ell.
\end{equation*}
Furthermore, we abbreviate the supercharges acting on a chain of length $N$ (resp. $\tilde N$) with periodic or special boundary conditions by $Q$ (resp. $\tilde Q$). Our aim is to apply the tic-tac-toe lemma, for which we divide the chain into a sublattice $S_1$ consisting of the sites $1,2,\dots,\ell+2$, and $S_2$ the remaining sites :
\begin{center}
 \begin{tikzpicture} 
    
     \draw[thick] (0,0) rectangle (1.8,0.3);
     \draw[thick,xshift=1.8cm, densely dotted] (0,0) rectangle (4,0.3);
         
     \draw (.9,.6) node{$\overset{\ell+2\,\text{sites}}{\overbrace{\hspace{1.7cm}}}$};
     \draw (.9,-.3) node{$S_1$};  
     \draw (3.8,.6) node{$\overset{\tilde N\,\text{sites}}{\overbrace{\hspace{3.9cm}}}$};
      \draw (3.8,-.3) node{$S_2$};
  \end{tikzpicture}
\end{center}
We may restrict the supercharge $Q$ to each of these sublattices, and find 
\begin{equation*}
  Q = Q_1+Q_2,\quad Q_i \equiv Q|_{S_i},\,,i=1,2.
\end{equation*}
Here and in the following, it is understood that the restricted supercharges respect the boundary conditions $(p)$ and $(s)$ for the chain of length $N$. One checks then that for both type of boundary conditions, the operators $Q_1$ and $Q_2$ are nilpotent and anticommute. This allows to consistently define the two-step cohomology $\coh_{21} \equiv \coh_{Q_2}(\coh_{Q_1}(V_N^{(p/s)}))$. In the course of its analysis, we will frequently encounter the state
\begin{equation*}
\ket{\chi} \equiv |0\underset{\ell \text{-cluster}}{\underbrace{1\dots1}}0\rangle. 
\end{equation*}

\paragraph{Structure of the cohomologies.}

Our first goal is to characterize the cohomology of the supercharge $Q_1$ :
\begin{proposition}
  Let $N > 2(\ell+1)$ then any non-trivial element of $\coh_{Q_1}$ with grade  $f$ has $f\geq \ell $, and can be represented by a state $|\psi\rangle \in  V_{N,f}$ of the form
  \begin{equation*}
    |\psi\rangle = |\chi\rangle \otimes |\psi'\rangle \quad \text{for some non-zero}\quad |\psi'\rangle \in V_{\tilde N,\tilde f},
  \end{equation*}
  with $\tilde N=N-\ell-2,\tilde f=f-\ell$.
  \begin{proof}
    Let us start by considering $\coh_{Q_1}$ for special boundary conditions $s=(c_1,c_N)$. Notice that since $Q_1$ acts only on the first $\ell+2$ sites we may solve the cohomology problem independently in any subspace of $V^{(s)}_{N,f}$ which contains all states with a $k$-cluster starting at site $\ell+3$ :
\begin{center}
 \begin{tikzpicture} 
     \begin{scope}
     \draw[thick,xshift=1.8cm] (0,0) rectangle (1.8,0.3);
    
     \draw[dotted] (3.6,0.3)--(5.15,.3);
     \draw[dotted] (3.6,0.)--(5.15,0.); 
     
     \draw(4.1,-.3) node{$\underset{k\text{-cluster}}{\underbrace{\hspace{0.75cm}}}$};
     \draw (2.7,.6) node{$\overset{\text{sites}\, 1,2,\dots,\ell+2}{\overbrace{\hspace{1.7cm}}}$};

     \draw (4.4,0.15) node {\small $1\cdots 10\cdots $};
     
     \draw (8,0.15) node {$\text{with} \quad k=0,\dots, \min(f,\ell)$}; 
     \end{scope}
     
  \end{tikzpicture}
\end{center}
Within each such subspace, finding $\coh_{Q_1}$ amounts to solving the cohomology problem for an open chain of length $\ell+2$ with boundary conditions $(c_1,c_{\ell+2}=\ell-k)$. Notice however, that this is true regardless of the boundary condition for the full chain if only if $N> 2(\ell+1)$. Recall from \cref{lem:special} that for chains of length $\ell+2$ the cohomology with special boundary conditions is non-trivial if and only if
\begin{equation}
  c_1 + c_{\ell+2} \geq \ell 
  \label{eqn:boundk}
\end{equation}
If this restriction is met, then the non-trivial elements have grade $\ell$ and can be represented up to a factor by the state $\ket{\chi}$.
For the full chain, we conclude that the cohomology problem has therefore a non-trivial solution if and only if
\begin{equation*}
  k\leq c_1 \quad \text{and} \quad f\geq \ell.
\end{equation*}
The bound on $k$ implies that the non-trivial elements of $\coh_{Q_1}$ at grade $f$ can be represented by states of the form $|\psi\rangle = |\chi\rangle \otimes |\psi'\rangle$ where $|\psi'\rangle$ is a non-zero state on $S_2$, subject to the same boundary conditions $(c_1,c_N)$ as on the full chain, and hence $|\psi'\rangle \in V^{(s)}_{\tilde N,\tilde f}$.

For periodic boundary conditions the reasoning is similar. First, notice that we may solve the cohomology problem for $Q_1$ independently in the subspaces of $V^{(p)}_{N,f}$ containing all states with a $k$-cluster starting at site $\ell+3$, and $m$-cluster ending at site $N$ :
  \begin{center}
 \begin{tikzpicture} 
     \begin{scope}
     \draw[thick,xshift=1.8cm] (0,0) rectangle (1.8,0.3);
    
     \draw[dotted] (0.3,0.3)--(1.8,.3);
     \draw[dotted] (0.3,0.)--(1.8,0.); 
    
     \draw[dotted] (3.6,0.3)--(5.15,.3);
     \draw[dotted] (3.6,0.)--(5.15,0.); 
     
     \draw (1.35,-.3) node{$\underset{m\text{-cluster}}{\underbrace{\hspace{.75cm}}}$};
     \draw(4.1,-.3) node{$\underset{k\text{-cluster}}{\underbrace{\hspace{0.75cm}}}$};
     \draw (2.7,.6) node{$\overset{\text{sites}\, 1,2,\dots,\ell+2}{\overbrace{\hspace{1.7cm}}}$};

     \draw (1.05,0.15) node {\small $\cdots 01\cdots 1$};
     \draw (4.4,0.15) node {\small $1\cdots 10\cdots $};
     \end{scope}
     
  \end{tikzpicture}
\end{center}
Given $k,m$ such that $k+m\leq f$ we need to solve the cohomology problem for an open chain with boundary conditions $(c_1=\ell-m,c_{\ell+2}=\ell-k)$. From \eqref{eqn:boundk} and the grade of the representative for the short chain we conclude thus
\begin{equation*}
  m+k\leq \ell \quad \text{and} \quad f\geq \ell.
\end{equation*}
  The first inequality leads to the conclusion that each non-zero element of $\coh_{Q_1}$ can thus be represented by $|\psi\rangle = |\chi\rangle \otimes |\psi'\rangle$ where $|\psi'\rangle \in V_{\tilde N,\tilde f}^{(p)}$ is subject to periodic boundary conditions.
  \end{proof}
  \label{prop:structureHQ1}
\end{proposition}

This proposition characterizes the representatives of all non-trivial elements in $\coh_{Q_1}$ (the statement can obviously be extended to representatives of the zero element if $|\psi'\rangle=0$ is chosen). Thus all states in the kernel of $Q_1$ which are not of the form as stated above are thus coboundaries. We use this observation to deduce two useful results.
\begin{corollary}
  Let $|\psi'\rangle \in V^{(p/s)}_{\tilde N}$ and $|\phi\rangle \in  V^{(p/s)}_{N}$ then the equation
  \begin{equation*}
    |\chi\rangle \otimes |\psi'\rangle = Q_1|\phi\rangle
  \end{equation*}
  has no solutions other than $|\psi'\rangle = 0$ and $Q_1|\phi\rangle = 0$.
  \label{corr:useful1}
\end{corollary}

\begin{corollary}
  Consider the chain of length $N$ with boundary conditions $(p)$ or $(s)$ with $s=(c_1,c_N)$. Let $|\psi'\rangle$ be a state of the Hilbert space $V_{\tilde N}$ with free conditions, whose first/last site is part of a $k$-cluster/m-cluster respectively, where $k+m > \ell$ for the choice $(p)$, and $k>c_1,m=c_N$ for choice $(s)$. Then there is a state $|\phi\rangle \in  V^{(p/s)}_{N}$ such that
  \begin{equation*}
    |\chi\rangle \otimes |\psi'\rangle = Q_1|\phi\rangle.
  \end{equation*}
  \label{corr:useful2}
\end{corollary}

\cref{prop:structureHQ1} suggests the introduction of the mapping
\begin{equation}
   G: V_{\tilde N}^{(p/s)} \to V_N^{(p/s)},\quad |\psi'\rangle \mapsto |\psi\rangle = |\chi\rangle \otimes |\psi'\rangle.
  \label{eqn:defg}
\end{equation}
which will play a central role in the following. It is quite obvious that $G$ maps any state $|\psi'\rangle$ into the kernel of $Q_1$ : $Q_1\left(G|\psi\rangle\right)=0$. In view of our aim to establish a relation between the cohomology of $\tilde Q$ and the two-step cohomology of the tic-tac-toe lemma, we also need the action of $Q_2$ on $G|\psi'\rangle$ which leads to the following interesting commutation relation :
\begin{lemma}
  Let $|\psi'\rangle\in V_{\tilde N}^{(p/s)}$ then there is a state $|\phi'\rangle \in V_N^{(p/s)}$ such that
  \begin{equation*}
    Q_2\left(G|\psi'\rangle\right) = (-1)^\ell \left(G\left(\tilde Q|\psi'\rangle\right)+Q_1|\phi'\rangle\right).
  \end{equation*}
  \begin{proof}
    Let us compute
    \begin{align*}
      Q_2\left(G|\psi'\rangle\right) & = Q_2\left(|\chi\rangle \otimes |\psi'\rangle\right) = (-1)^{\ell}|\chi\rangle \otimes Q_2|\psi'\rangle\\
       & = (-1)^{\ell}\left(|\chi\rangle \otimes P Q_2|\psi'\rangle+|\chi\rangle \otimes (1-P) Q_2|\psi'\rangle
\right).
    \end{align*}
    Here $P$ denotes the projector on $V_{\tilde N}^{(p/s)}$. We know that $\tilde Q = PQ_2$. Furthermore, $(1-P)Q_2|\psi'\rangle \notin V_{\tilde N}^{(p/s)}$ by construction, but fulfils the requirements of \cref{corr:useful2}. We infer thus the existence of a state $|\phi'\rangle \in V_N^{(p/s)}$ such that $|\chi\rangle \otimes (1-P) Q_2|\psi'\rangle = Q_1|\phi'\rangle$. Replacing this into the expression here above proves the claim.
   \end{proof}
   \label{lem:Q2onG}
\end{lemma}

\bigskip

Our next step consists of understanding the structure of the non-trivial elements of the two-step cohomology group $\mathcal H_{21}$ with the help of the results about $\coh_{Q_1}$.

\begin{proposition}
  Let $N > 2(\ell+1)$ then any
  element of $\coh_{21}$ and definite grade $f$ has $f\geq \ell $, and can be represented by a state $|\psi\rangle \in  V^{(p/s)}_{N,f}$ of the form
  \begin{equation*}
    |\psi\rangle = G|\psi'\rangle, \quad \tilde Q|\psi'\rangle = 0,
  \end{equation*}
  for some non-zero $|\psi'\rangle \in V^{(p/s)}_{\tilde N,\tilde f}$ with $\tilde N=N-\ell-2,\,\tilde f=f-\ell$.  \begin{proof}
    Any representative $|\psi\rangle$ of an element in $\coh_{21}$ needs to solve the two equations \cref{eqn:CondRepH21}. First, according to \cref{eqn:CondRepH21First}, it needs to be in the kernel of $Q_1$. We know from \cref{prop:structureHQ1} that the non-trivial elements of the latter are (up to a $Q_1$-coboundary which can be chosen zero without loss of generality) given by $|\psi\rangle = G|\psi'\rangle$ with a non-zero state $|\psi'\rangle \in V^{(p/s)}_{\tilde N}$ (for representatives of the trivial element, it is sufficient to choose $|\psi'\rangle =0$). Second, \cref{eqn:CondRepH21Second} implies that $Q_2(G|\psi'\rangle) = Q_1|\phi\rangle$ for some $|\phi\rangle$. We use \cref{lem:Q2onG} to infer the existence of a state $|\phi'\rangle$ such that
  \begin{equation*}
    (-1)^{\ell}\left(|\chi\rangle \otimes \tilde Q|\psi'\rangle+Q_1|\phi'\rangle\right) = Q_1 |\phi\rangle 
  \end{equation*}
  for some $|\phi\rangle \in V_N^{(p/s)}$. Finally, by virtue of \cref{corr:useful1} this equation holds if and only if $\tilde Q|\psi'\rangle=0$, what finishes the proof.
  \end{proof} 
  \label{prop:structureRepH21}
\end{proposition}
This proposition has at least two important consequences. The first one is a quite straightforward consequence of the tic-tac-toe lemma :
\begin{corollary}\label{cor:iso}
For $N>2(\ell+1)$ the cohomology of $Q$ is isomorphic to $\coh_{21}$.
\end{corollary}
\begin{proof}
This is an immediate consequence of \cref{prop:ttt,prop:structureRepH21} since all non-trivial elements of $\coh_{Q_1}$ can be represented by states which have $f_1=\ell$ particles on $S_1$.
\qedhere
\end{proof}

\paragraph{Cohomology isomorphism.} This proposition reveals an interesting insight into the structure of the representatives of $\coh_{21}$ at $N$ sites: they can be constructed from the elements of the kernel of $\tilde Q$ at $\tilde N=N-(\ell+2)$ sites. This observation suggests to lift $g$ to a mapping between cohomology groups
\begin{equation*}
  G^\sharp : \coh_{\tilde Q}(V_{\tilde N}^{(p/s)})\to \coh_{21}, \quad [|\psi'\rangle]_{\coh_{\tilde Q}} \mapsto [G|\psi'\rangle]_{\coh_{21}}.
\end{equation*}
For this mapping to be well-defined, we need to verify that $G$ sends cocycles onto cocycles, and coboundaries onto coboundaries :
\begin{itemize}
\item Let us start with a $\tilde Q$-cocycle $|\psi'\rangle$ : $\tilde Q|\psi'\rangle=0$. We need to  check that $|\psi\rangle  = G|\psi'\rangle$ obeys \eqref{eqn:CondRepH21}. First, it is trivial to see that $Q_1|\psi\rangle = 0$. Second, from \cref{lem:Q2onG} it follows that there is a state $|\phi'\rangle$ such that $Q_2 |\psi\rangle = Q_1 |\phi'\rangle$. This implies that $|\psi\rangle$ is a cocycle in the sense of the two-step cohomology.
\item Next, suppose that $|\psi'\rangle = \tilde Q|\psi''\rangle$ is a $\tilde Q$-coboundary. In this case, adapting slightly the logic of the derivation of \cref{lem:Q2onG} it is not difficult to see that there is a state $|\phi''\rangle$ such that
\begin{equation*}
  G|\psi'\rangle = Q_2((-1)^\ell g|\psi''\rangle)+Q_1(-|\phi''\rangle).
\end{equation*}
Since $Q_1(G|\psi''\rangle)=0$ trivially, the right-hand side of this equation is a coboundary in the sense of the two-step cohomology, and represents thus the trivial element of $\coh_{21}$.
\end{itemize}

\begin{proposition}
  For $N > 2(\ell+1)$ the mapping $G^\sharp: \coh_Q(V_{\tilde N}^{(p/s)}) \to \coh_{21}$ is an isomorphism.
  \begin{proof}
    We show that $G^\sharp$ is an isomorphism by establishing surjectivity and injectivity. We prove these by analyzing the action of $G$ on representatives. Notice that since we showed that $G^\sharp$ is well-defined, we may choose representatives up to coboundaries, and thus work with the convenient forms given in \cref{prop:structureHQ1,prop:structureRepH21}.
    
    First, surjectivity of $G^\sharp$ states that each element of $\mathcal H_{21}(V_N^{(p/s)})$ has a pre-image in $H_Q(V_{\tilde N}^{(p/s)})$. On the level of representatives this means that for any $|\psi\rangle$ representing an element of $\coh_{21}$ there is $|\psi'\rangle \in \ker Q_1$ such that $|\psi\rangle = G|\psi'\rangle$. Without loss of generality, we may choose $|\psi\rangle$ as in \cref{prop:structureRepH21}, which proves the statement.
    
    Second, injectivity of $G^\sharp$ states that $\ker G^\sharp = [0]_{\coh_{\tilde Q}}$. It follows from \eqref{eqn:ZeroInH21} that at the level of representatives, this is equivalent to say that
    \begin{equation*}
      G|\psi'\rangle = Q_2|\alpha\rangle + Q_1 |\beta\rangle \quad \text{with} \quad Q_1|\alpha\rangle = 0
    \end{equation*}
    implies $|\psi'\rangle \in \im \tilde Q$. Since $Q_1|\alpha\rangle = 0$ we may without loss of generality assume that $|\alpha\rangle = G|\alpha'\rangle$ for some $|\alpha'\rangle \in V_{\tilde N}^{(p/s)}$. Using \cref{lem:Q2onG} we infer the existence of a state $|\alpha''\rangle$ such that
    \begin{equation*}
    G|\psi'\rangle = G \left((-1)^\ell \tilde Q|\alpha'\rangle\right) + Q_1 \left((-1)^\ell|\alpha''\rangle+|\beta\rangle\right).
    \end{equation*}
    From \cref{corr:useful1} it follows thus $|\psi'\rangle = (-1)^\ell \tilde Q|\alpha'\rangle \in \im \tilde Q$, what concludes the prove of injectivity.
    \end{proof}
    \label{prop:CohomIsomorph}
\end{proposition}

\paragraph{Cut (and paste).} We are now in a position to prove the \textit{cut (and paste)} cohomology isomorphism :
\begin{proof} [Proof of \cref{thm:cutandpaste}]

First, \cref{cor:iso} states that
\begin{equation*}
  \coh_{Q} \simeq \coh_{21}
\end{equation*}
via the tic-tac-toe isomorphism. Second, we proved in \cref{prop:CohomIsomorph} that
\begin{equation*}
\coh_{21} \simeq \coh_{\tilde Q}.
\end{equation*}
provided that $N>2(\ell+1)$ via $G^\sharp$. These relations hold for both boundary conditions $(p)$ and $(s)$. Since the isomorphism relation is transitive, we conclude that $\coh_{\tilde Q}\simeq \coh_{Q}$ for $N>2(\ell+1)$. Furthermore, since the map $G$ adds $\ell$ particles to a state, we conclude that the non-trivial element of $\coh_{Q}$ have $\ell$ more particles than the non-trivial elements of $\coh_{\tilde Q}$, and therefore
\begin{equation*}
  \coh_{\tilde Q}^f \simeq \coh_{Q}^{f+\ell}.
\end{equation*}

\qedhere
\end{proof}

The theorem implies in particular $\dim \coh_{\tilde Q}^f = \dim \coh_{Q}^{f+\ell}$ which we are going to apply repeatedly in the next sections.

\subsection{Closed chains}
\label{sec:closedchain}

Here we present the proof of \cref{prop:closedchain}.
\begin{proof}[Proof of \cref{prop:closedchain}.]
To prove the statement about the zero-energy ground states we solve the corresponding cohomology problem and use the fact that zero-energy states are in one-to-one correspondence with cohomology elements.

For $n=0$ the solution to the cohomology problem can be found in \cref{lem:periodicshort1}. For $n>1$ and $0 \leq p \leq \ell-1$ we first use \cref{thm:cutandpaste} $n-1$ times to reduce the cohomology problem of an $N=n(\ell+2)+p+1$ site closed chain to the cohomology problem for a chain with periodic boundary conditions and $N'=\ell+p+3$ sites. Its solution can be found in \cref{lem:periodicshort2}. We thus find that $\coh_Q$ is one-dimensional and non-trivial only at grade $N'-3+(n-1)\ell = n \ell+p$.

Furthermore, for $n>1$ and $p=\ell,\ell+1$ we use \cref{thm:cutandpaste} $n$ times to reduce the cohomology problem of an $N=n(\ell+2)+p+1$ site closed chain to the cohomology problem of an $N'=p+1$ site closed chain. \cref{lem:periodicshort1} provides its solution. For $p=\ell$ we again find that $\coh_Q$ is one-dimensional, and non-trivial only at grade $N'-1+n\ell = n \ell+p$. Finally, for $p=\ell+1$ we find that $\coh_Q$ has dimension $\ell+1$ and is non-trivial only at grade $N'-1+n\ell = (n+1) \ell$. 

\qedhere
\end{proof}

\subsection{Open chains}\label{sec:specialchain}

Here we present the proof of proposition \cref{prop:specialchain}.
\begin{proof}[Proof of \cref{prop:specialchain}]
As above, we use the fact that zero-energy states are in one-to-one correspondence with cohomology elements. This allows to prove the statement about the zero-energy ground states by solving the corresponding cohomology problem. 

For $0 \leq p \leq \ell-1$ we first use proposition \cref{thm:cutandpaste} $n-1$ times to reduce the cohomology problem of an $N=n(\ell+2)+p+1$ site open chain with special boundary conditions to the cohomology problem of an $N'=\ell+p+3$ site open chain with the same special boundary conditions. The solution to this problem can be found in \cref{lem:specialshort2}. We thus find that $\coh_Q$ has dimension one when (i) $c_1,c_N \geq p$ and $c_1+c_N \leq \ell+p$ for $0 \leq p \leq \ell-1$. In both cases, it is non-trivial only at grade $N'-3+(n-1)\ell = n \ell+p$. Furthermore, it is also one-dimensional when (ii) $c_1,c_N < p$ and $c_1+c_N \geq p-1$ for $1 \leq p \leq \ell-1$, and non-trivial only at grade $N'-4+(n-1)\ell = n \ell+p-1$.

For $p=\ell,\ell+1$ we use \cref{thm:cutandpaste} $n$ times to reduce the cohomology problem of an $N=n(\ell+2)+p+1$ site open chain with special boundary conditions to the cohomology problem of an $N'=p+1$ site open chain with the same special boundary conditions. We deduce the solution to this problem from \cref{lem:special}. Again, we find that $\coh_Q$ is one-dimensional when (i) $c_1,c_N \geq p$ and $c_1+c_N \leq \ell+p$ for $p=\ell$, and non-trivial only at grade $N'-1+n\ell = n \ell+p$. Furthermore, it is one-dimensional for (ii) $c_1,c_N < p$ and $c_1+c_N \geq p-1$ for $p =\ell, \ell+1$. In this case, it is non-trivial only at grade $N'-2+n\ell = n \ell+p-1$.

\qedhere
\end{proof}

\section*{Acknowledgements} LH acknowledges funding by the John Templeton Foundation and a DOE early career award. The work of LH was in part performed at the Aspen Center for Physics, which is supported by National Science Foundation grant PHY-1066293. CH is supported by the Belgian Interuniversity 
Attraction Poles Program P7/18 through the network DYGEST (Dynamics, Geometry and 
Statistical Physics). He acknowledges the Stanford Institute for Theoretical Physics where part of this work was done for hospitality.

\appendix

\section{Alternative cohomology proofs}\label{app:altproofs}
In this appendix, we present an alternative proof of the cohomology result for the open chain with free boundary conditions and general length $N$ and for the closed chain with periodic boundary conditions and length $N=n(\ell+2)$. These results are contained in \cref{prop:specialchain} with special boundary conditions $s=(\ell,\ell)$ and \cref{prop:closedchain}, respectively. The arguments are relatively simple compared to the proofs of the full cohomology results for closed and open chains. This is the reason we include them here.

\paragraph{Single-site result.} Our strategy relies on the tic-tac-toe lemma. We shall often take the sublattice, $S_1$, to be a single site or a collection of isolated sites that are more than $\ell$ sites apart. For this reason we will use the cohomology result for a single site again and again, so let us state it here. The cohomology of a single site is trivial if the site can be both empty and occupied : The empty state is not in the kernel of $Q_1$ and the occupied state is in the image of $Q_1$. Therefore both states belong to the trivial equivalence class. However, the cohomology is non-trivial if the site has to be empty due to either explicit boundary conditions or effective boundary conditions derived from the occupation of neighboring sites belonging to sublattice $S_2$. This is because in this case the empty state is in the kernel of $Q_1$ and it is clearly not in the image of $Q_1$.

\subsection{Open chains with free boundary conditions}\label{sec:freebc}
We address the cohomology problem for an open chain with free boundary conditions which means that there are no restrictions on the length of connected particle clusters containing the first or last site (other than the exclusion constraints). We will prove that the cohomology of the $M_{\ell}$ model on an open chain of length $N=n(\ell+2) + p + 1,\,n\geq 0, \,p=0,\dots,\ell+1,$ with free boundary conditions has dimension one for $p=\ell+1$ and $p=\ell$ and has dimension zero otherwise. It is non-trivial only at grade $f=(n+1)\ell$ in both cases.

\begin{proof}
We first define the sublattices we use to apply the tic-tac-toe lemma. A convenient choice for $S_1$ is given by
\begin{equation*}
  S_1=\{j(\ell+2)+1\}_{j=0}^n = \{1,\ell+3,\dots, 2\ell+5,\dots, n(\ell+2)+1\}
\end{equation*}
Let us illustrate this with an example where $n=3$:
\begin{center}
  \begin{tikzpicture}
     \foreach \x in {0,1.8,3.6,5.4}
     {
     \draw[thick,xshift=\x cm] (0,0) rectangle (.3,.3);
     }
     \foreach \x in {0,1.8,3.6}
     {
     \draw[dotted,xshift=\x cm] (0.3,0.3)--(1.8,.3);
     \draw[dotted,xshift=\x cm] (0.3,0.)--(1.8,0.); 
     \draw[xshift=\x cm] (1.05,-.3) node{$\underset{\ell+1\,\text{sites}}{\underbrace{\hspace{1.5cm}}}$};
     }
     \draw[dotted] (5.7,.3)--(6.6,.3) -- (6.6,0) --(5.7,0);
     \draw(6.15,-.3) node{$\underset{p\,\text{sites}}{\underbrace{\hspace{.9cm}}}$};

  \end{tikzpicture}
\end{center}
Here the squares represent the sites of $S_1$. They are equally-spaced and isolated, separated by $\ell+1$ sites belonging to $S_2$. This choice is a natural generalization of the 3-rule which was used to study the $M_1$ model on various graphs \cite{fendley:05_2}. 

In order to find $\coh_{Q_1}$ we use the cohomology result for the single site (see end of \cref{sec:model}). We find that $\coh_{Q_1}$ is non-trivial if all sites of $S_1$ are empty and cannot be occupied. The states for which this is true take the form:
\begin{equation}
    |\bm{0}\,\underset{\ell\text{-cluster}}{\underbrace{1\cdots 1}} {0}\,\bm{0}\,\underset{\ell\text{-cluster}}{\underbrace{1\cdots 1}}{0}\, \bm{0}\cdots \bm{0}\,\underset{\ell\text{-cluster}}{\underbrace{1\cdots 1}} {0}\,\bm{0}\rangle \otimes |\psi\rangle. \nonumber
  \end{equation}
Here we printed in bold the sites of $S_1$. Furthermore, $|\psi\rangle$ is a state with $p$ sites. To see this we start at the first site. The first $S_1$ site, i.e. site 1, has to be empty if an only if it is adjacent to a cluster of $\ell$ particles that starts on site 2. The exclusion rule then implies that site $\ell+2$ has to be empty. Next, look at the second site in $S_1$, i.e. site $1+(\ell+2)$, for which we also require that it be empty. We see that the problem is identical to the one for the first site and thus a cluster of $\ell$ particles has to start on site $2+(\ell+2)$. By recursion we arrive at the last site in $S_1$, i.e. site $n(\ell+2)+1$. Again it has to be adjacent to an $\ell$-cluster. For $p< \ell$ the problem has no solution and therefore the cohomology $\coh_{Q_1}$ is empty. Conversely, if $p=\ell,\ell+1$ the solution is unique:
\begin{equation}
    |\psi\rangle =
    \begin{cases}
       |\underset{\ell\text{-cluster}}{\underbrace{1\cdots 1}}0\rangle ,& p = \ell+1,\\
       |\underset{\ell\text{-cluster}}{\underbrace{1\cdots 1}}\rangle,& p = \ell.
    \end{cases} \nonumber
\end{equation}
  
We conclude that the cohomology $\coh_{Q_1}$ of $Q_1$ for a chain of length $N=n(\ell+2) + p + 1$ has dimension one if $p=\ell$ or $p=\ell+1$, and zero otherwise. Since the dimension of $\coh_{Q_1}$ is either zero or one, the second step, i.e. the computation of the cohomology $\coh_{Q_2}(\coh_{Q_1})$ of $Q_2$ acting within $\coh_{Q_1}$, is trivial. The representatives have $f=(n+1)\ell$ fermions, and so do the ground states.

\qedhere

\end{proof}

\subsection{Closed chains with $N=n(\ell+2)$}
\label{app:pbc}
In this appendix we present an alternative proof of the cohomology computation for periodic chains where the number of sites is a multiple of $\ell+2$. From \cref{prop:closedchain} we know that the Hamiltonian of the periodic chain with $N=n(\ell+2)$ sites, $n \geq 1$, has $\ell+1$ zero-energy ground states with fermion number $f=n\ell$.

\begin{proof}
  As in the appendix on open chains, \cref{sec:freebc}, let us choose the subset of lattice sites $S_1 =\{k(\ell+2)+1\}_{k=0}^n$. Any two consecutive sites of $S_1$ are separated by $\ell+1$ sites belonging to $S_2$, including the last one and the first one. Using the cohomology result for a single site, we obtain once more that all elements of $\coh_{Q_1}$ can be represented by simple basis states which are such that all sites on $S_1$ are empty and cannot be occupied. Hence, the $k$-th site in $S_1$ must have a cluster of $a_k$ particles to its left, and a cluster of $b_k$ particles to its right, such that the sum of their lengths is at least $\ell$. 
  Moreover, the exclusion rule stipulates that between two consecutive $S_1$ sites there can be most $\ell$ particles, which implies
  \begin{equation}
    b_{k} + a_{k+1} \leq \ell. \nonumber
  \end{equation}
  Here, it is understood that for $k=n$ we write $a_{n+1} = a_0$. Upon elimination of the $b$'s we conclude from these inequalities that the $a$'s have the property $a_k \geq a_{k+1}$ for all $k$, and hence
  \begin{equation*}
    a_0 \geq a_1 \geq \cdots \geq a_{n} \geq a_0. 
  \end{equation*}
  Hence, $a_k=a_0$ for all $k=0,\dots,n$. It follows immediately that $b_k=\ell-a_0$ for all $k=0,\dots,n$. The particle arrangements around the sites of $S_1$ look therefore all alike for a given value of $a_0$. The basis of $\mathcal H_{Q_1}$ can thus be represented by the following $\ell+1$ states
  \begin{equation}
    |\bm 0 \underset{(\ell-a_0)\text{-cluster}}{\underbrace{1\cdots 1}} 0 \underset{a_0\text{-cluster}}{\underbrace{1\cdots 1}}\rangle^{\otimes (n+1)}, \quad a_0=0,\dots,\ell.
    \label{eqn:perreps}
  \end{equation}
  which have all fermion number $f=(n+1)\ell$.
    
  The computation of the cohomology of $Q_2$ acting on $\coh_{Q_1}$ is quite simple. One checks easily that all the representatives in \eqref{eqn:perreps} are in the kernel of $Q_2$. Conversely, $\coh_{Q_1}$ does not contain any element which leads to one of these representatives upon action with $Q_2$. We conclude that the states given in \eqref{eqn:perreps} constitute representatives for the basis of $\coh_{Q_2}(\coh_{Q_1})$. Hence $\dim \coh_Q=\dim \coh_{Q_2}(\coh_{Q_1}) = \ell+1$. Furthermore, the cohomology is non-trivial only at grade $f=(n+1)\ell$ as follows from the representatives \eqref{eqn:perreps}.

\qedhere
\end{proof}


\section{Short chains}\label{app:short}
In this appendix we solve the cohomology problem of short chains, $1\leq N \leq 2\ell+2$, and prove the propositions stated in \cref{sec:shortlems}. We restate them below for convenience. Furthermore, we use the single-site result stated in \cref{app:altproofs}, and the following proposition, which allows to infer representatives for the cohomology of $Q$ from the two-step cohomology :
\begin{proposition}
\label{lem:repr}
Let $|\psi\rangle$ be a representative of a non-zero element of $\coh_{Q_2}(\coh_{Q_1})$  with no particles on $S_1$ (i.e. $f_1=0$), then $|\psi\rangle$ is also a representative of a non-zero element of $\coh_Q$.
\begin{proof}
  Suppose the opposite : $|\psi\rangle = Q |\psi'\rangle = Q_1 |\psi'\rangle+Q_2 |\psi'\rangle$ for some $|\psi'\rangle$. As $|\psi\rangle$ has no particles on $S_1$ we have $Q_1 |\psi'\rangle=0$, and thus $|\psi\rangle = Q_2 |\psi'\rangle$. Hence according to \eqref{eqn:ZeroInH21} the state $|\psi\rangle$ represents the trivial equivalence class $\coh_{Q_2}(\coh_{Q_1})$, which is a contradiction.
\end{proof}
\end{proposition}

\subsection{Special boundary conditions and $1\leq N\leq \ell+2$}
\label{app:specialshort1}

In this section, we consider special boundary conditions $s=(c_1,c_N)$ where $0\leq c_1,c_N\leq \ell$. We may encounter the special case where $N \leq \min(\ell, c_1,c_N)$. In this case the length of any particle cluster is bounded only by $N$, there are no additional constraints. Furthermore, we point out that when $c_1 \geq N$, but $c_N <N$, the maximal length of the cluster that starts on the first site is actually $N-1$. Therefore, $c_1$ is effectively $N-1$ and thus the boundary conditions $s=(c_1 \geq N , c_N <N)$, and equivalently $s=(c_1 < N , c_N \geq N)$, do not make much sense in practice. Therefore, we will only consider the cases $0\leq c_1,c_N \leq \min(\ell, N-1)$ and the unconstrained case where the length of any particle cluster is bounded only by $N$. The result of \cref{lem:special} is illustrated in \cref{fig:specialshort1}.

\begin{repproposition}{lem:special}
The cohomology, $\coh_Q (V^{(s)}_N)$, of a chain of length $1\leq N\leq \ell+2$ with special boundary conditions given by $s=(c_1,c_N)$ with $0\leq c_1,c_N \leq min(\ell, N-1)$ has dimension one when (i) $c_1=c_N=N-1$ for $1\leq N\leq \ell+1$ and when (ii) $c_1,c_N \leq N-2$ and $c_1+c_N \geq N-2$ for $2 \leq N \leq \ell+2$ and dimension zero otherwise. The non-trivial elements can be represented by (i) $\ket{01\dots1}$ with $N-1$ particles and (ii) $\ket{01\dots10}$ with $N-2$ particles, respectively. The cohomology, $\coh_Q (V^{(s)}_N)$, of a chain of length $N$ in the unconstrained case, i.e. where the length of any particle cluster is bounded only by $N$, is trivial.
\end{repproposition}

\begin{figure}[ht]
\begin{center}
\includegraphics[width = 0.45\textwidth]{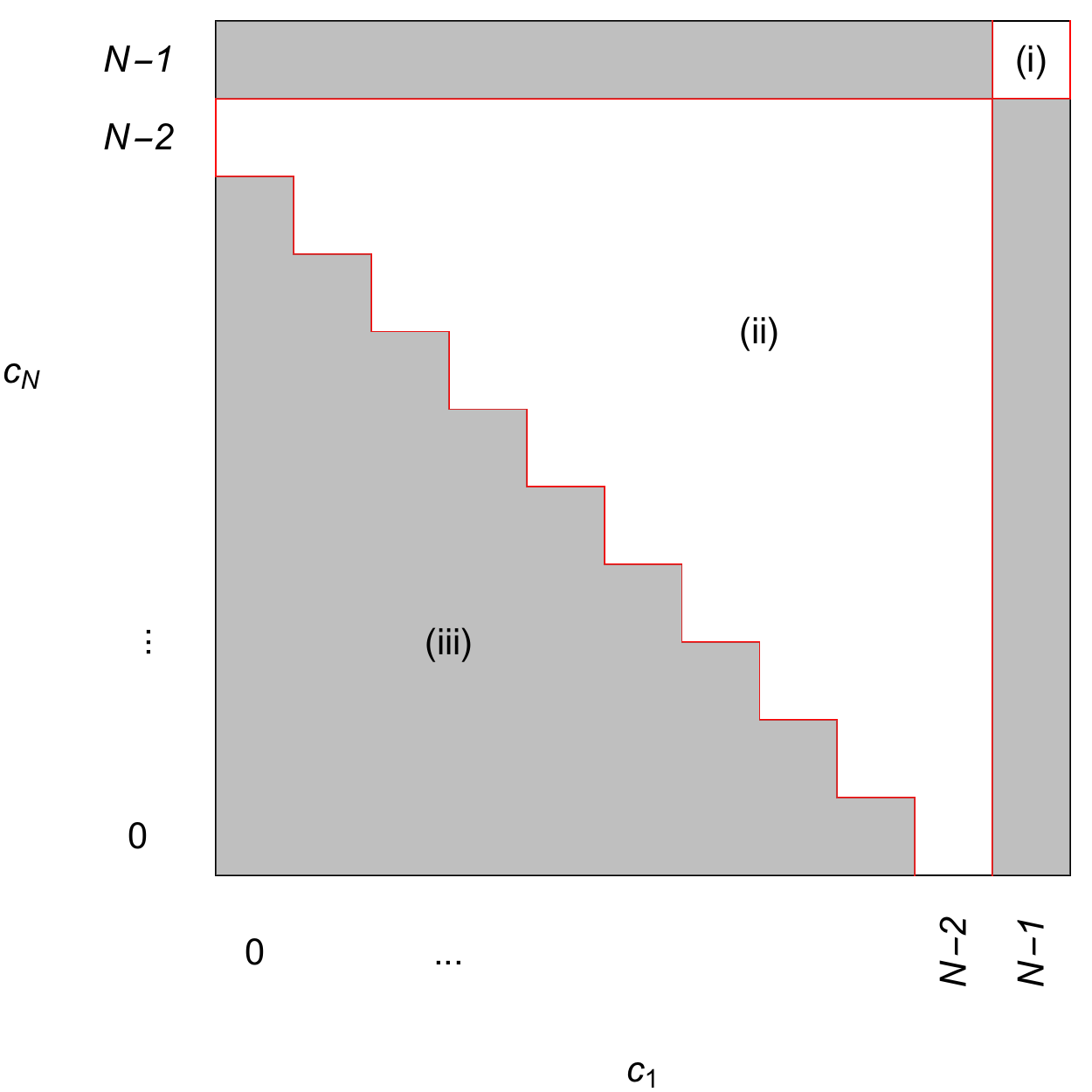}
\includegraphics[width = 0.45\textwidth]{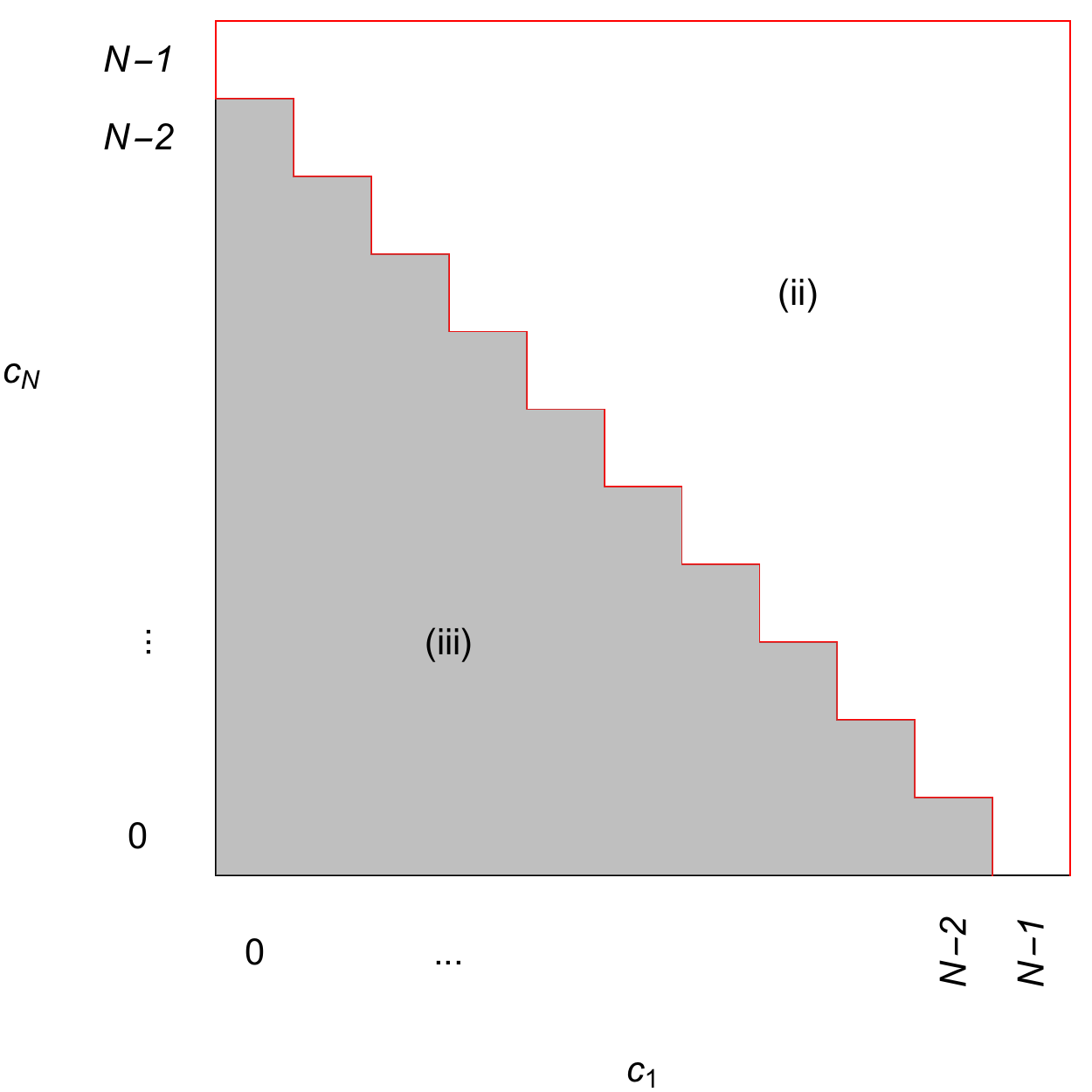}
\caption{We illustrate the structure of the cohomology for short chains with special boundary conditions. The regions corresponding to the boundary conditions for which the cohomology is trivial are gray shaded, for the boundary conditions corresponding to the white regions the dimension of the cohomology is one. The diagram on the left corresponds to $1 \leq N \leq \ell+1$, the diagram on the right is for $N=\ell+2$. The labels correspond to different cases we consider in the proposition and in the proof : (i) $c_1=c_N=N-1$, (ii) $c_1,c_N \leq N-2$ and $c_1+c_N \geq N-2$, (iii) $c_1,c_N \leq N-2$ and $c_1+c_N < N-2$. \label{fig:specialshort1}}
\end{center}
\end{figure}

In the proof we will use two choices for the sublattice, $S_1$: (1) $S_1$ contains the first site, (2) $S_1$ contains the first and the last site of the chain, in both cases all other sites belong to $S_2$. We can illustrate these choices as follows:
\begin{center}
  \begin{tikzpicture}      
     \node at (-0.8,0) {(1)};
     \foreach \x in {0}
     {
     \draw[thick,xshift=\x cm] (0,0) rectangle (.3,.3);
     }
     \foreach \x in {0}
     {
     \draw[dotted,xshift=\x cm] (0.3,0.3)--(1.8,.3);
     \draw[dotted,xshift=\x cm] (0.3,0.)--(1.8,0.); 
     \draw[dotted,xshift=\x cm] (1.8,0.)--(1.8,0.3); 
     \draw[xshift=\x cm] (1.05,-.3) node{$\underset{N-1\,\text{sites}}{\underbrace{\hspace{1.5cm}}}$};
     }
     \node at (2.8,0) {(2)};
     \foreach \x in {3.6,5.4}
     {
     \draw[thick,xshift=\x cm] (0,0) rectangle (.3,.3);
     }
     \foreach \x in {3.6}
     {
     \draw[dotted,xshift=\x cm] (0.3,0.3)--(1.8,.3);
     \draw[dotted,xshift=\x cm] (0.3,0.)--(1.8,0.); 
     \draw[xshift=\x cm] (1.05,-.3) node{$\underset{N-2\,\text{sites}}{\underbrace{\hspace{1.5cm}}}$};
     }

  \end{tikzpicture}
\end{center}
Here the squares represent the sites of $S_1$ and $S_2$ contains (1) $N-1$ or (2) $N-2$ consecutive sites.

\begin{proof}
We first solve the cohomology problem for the unconstrained case. We take sublattice choice (1): the first site belongs to $S_1$ and all other sites belong to $S_2$. Using the cohomology result for the single site it is easy to see that the cohomology of $Q_1$ is non-trivial if and only if the $S_1$ site is adjacent to a cluster of length at least $N$. This is clearly not possible since the length of $S_2$ is $N-1$. We thus find that $\coh_{Q_1}$, and therefore also $\coh_{Q_2}(\coh_{Q_1})$, is trivial, which implies that $\coh_Q$ is trivial.

Next, we solve the cohomology problem for $c_1=N-1$ and $0 \leq c_N \leq N-1$. Note that this requires that $N\leq \ell+1$ since $c_1 \leq \ell$. Again we take sublattice choice (1). The cohomology of $Q_1$ is non-trivial if and only if the $S_1$ site is adjacent to a cluster of length $N-1$. This is only possible when $c_N=N-1$. We thus find that $\coh_{Q_1}$, and therefore also $\coh_{Q_2}(\coh_{Q_1})$, is trivial for $c_1=N-1$ and $c_N < N-1$ and has dimension one when $c_1=c_N=N-1$. The non-trivial element has all $S_2$ sites occupied:
\beq
\ket{\psi}=\ket{\bm{0} \underset{N-1}{\underbrace{1\dots1}}},
\nonumber
\eeq
and thus $N-1$ particles. Similarly, we find that $\coh_{Q_2}(\coh_{Q_1})$ is trivial for $c_1 < N-1$ and $c_N=N-1$. It follows that $\coh_{Q}$ is trivial for $c_1=N-1$ and $c_N < N-1$ and for $c_1 < N-1$ and $c_N =N-1$ and it has dimension one when $c_1=c_N=N-1$. The non-trivial element can be written as $\ket{\psi}$ above using \cref{lem:repr} and has $N-1$ particles. This corresponds to case i) in \cref{lem:special} and in \cref{fig:specialshort1}. 

Finally, we turn to the case $c_1,c_N \leq N-2$. This requires that $N\geq 2$ since $c_1,c_N \geq 0$. We now take sublattice choice (2): $S_1$ contains the first and the last site and all other sites belong to $S_2$. We first note that the cohomology of $Q_1$ is equivalent to the cohomology of two independent single sites. To see this, note that the sites would not be independent if the occupancy of one site influences the constraints on the occupancy of the other site. This can only happen if there is a connected particle cluster on $S_2$ that is adjacent to both $S_1$ sites. For this particular case, this implies that all $N-2$ sites of $S_2$ are occupied. However, since $c_1,c_N \leq N-2$ we immediately find that both $S_1$ sites have to be empty in this case. We conclude that the constraints on the $S_1$ sites depend only on the occupancy of the $S_2$ sites and which means that the sites are independent. Thus using again the cohomology result of the single site, we find that the cohomology of $Q_1$ is non-trivial provided that a cluster of length $a \geq c_1$ starts on site 2 and a cluster of length $b\geq c_N$ ends on site $N-1$. We now distinguish two cases (see also the labels in \cref{fig:specialshort1}):
\begin{itemize}
\item[ii)] $c_1+c_N \geq N-2$,
\item[iii)] $c_1+c_N < N-2$.
\end{itemize}
For case ii) $\coh_{Q_1}$ has a unique non-trivial representative:
\beq
\ket{\bm{0} \underset{N-2}{\underbrace{1\dots1}} \bm{0}},
\nonumber
\eeq
this element has all $S_2$ sites occupied and thus $N-2$ particles. For case iii) all the non-trivial representatives of $\coh_{Q_1}$ can be written as:
\beq
\ket{\bm{0} \underset{c_1}{\underbrace{1\dots1}}} \otimes \ket{\psi'} \otimes \ket{\underset{c_N}{\underbrace{1\dots1}}\bm{0}},
\nonumber
\eeq
where $\ket{\psi'}$ is any state on the middle sites ($c_1+2, c_1+3, \dots, N-1-c_N$).

We now turn to the computation of $\coh_{Q_2}(\coh_{Q_1})$. For case ii) we immediately find that the dimension of $\coh_{Q_2}(\coh_{Q_1})$ is one. Using \cref{prop:ttt} we thus find that $\coh_Q$ has dimension one when $c_1,c_N \leq N-2$ and $c_1+c_N \geq N-2$. The non-trivial element can be written as $\ket{0 1 \dots 1 0}$ using \cref{lem:repr} and has $N-2$ particles. This corresponds to case ii) in \cref{lem:special} and in \cref{fig:specialshort1}. For case iii) we see that within $\coh_{Q_1}$ $Q_2$ effectively acts on a chain of length $N'=N-2-c_1-c_N$ where the length of the particle clusters is unconstrained. We have solved this case in the first paragraph of this proof. We thus find that $\coh_{Q_2}(\coh_{Q_1})$ is trivial, and therefore also $\coh_Q$, when $c_1,c_N \leq N-2$ and $c_1+c_N < N-2$.

\qedhere
\end{proof}

\subsection{Special boundary conditions and $\ell+ 3 \leq N\leq 2\ell+2$}

\begin{repproposition}{lem:specialshort2}
The cohomology, $\coh_Q (V^{(s)}_N)$, of a chain of length $\ell+ 3 \leq N\leq 2\ell+2$ with special boundary conditions given by $s=(c_1,c_N)$ with $0\leq c_1,c_N \leq \ell$ has dimension one when (i) $c_1,c_N \geq N-\ell-3$ and $c_1+c_N < N-2$ and when (ii) $c_1,c_N < N-\ell-3$ and $c_1+c_N \geq N-\ell-4$ for $\ell+3 < N \leq 2\ell+2$. Otherwise, its dimension is zero. The number of particles of the non-trivial elements is (i) $N-3$ and (ii) $N-4$, respectively.
\end{repproposition}

\begin{figure}[ht]
\begin{center}
\includegraphics[width = 0.65\textwidth]{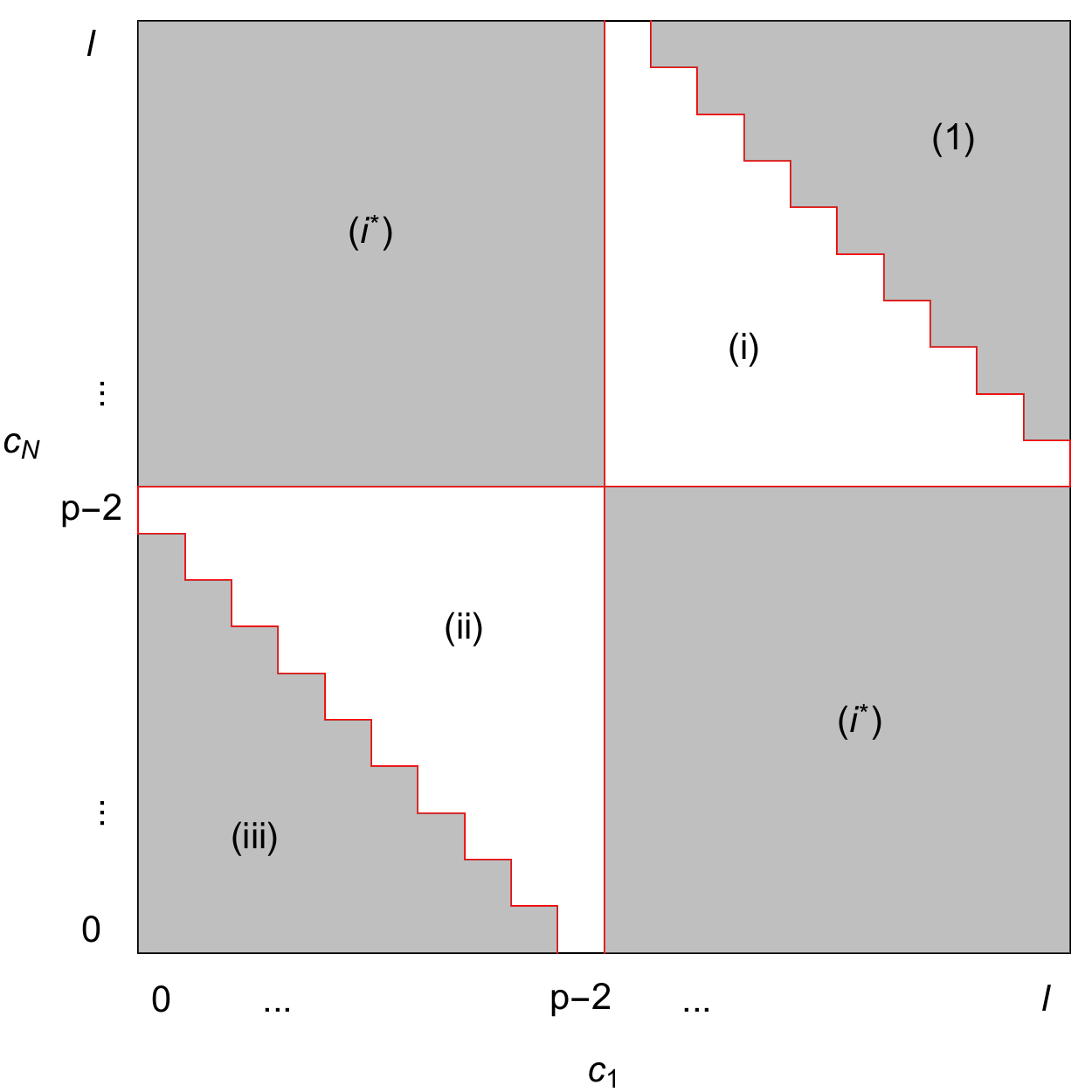}
\caption{We illustrate the structure of the cohomology for short chains with special boundary conditions. The length of the chain is $N=\ell+p+2$ with $1 \leq p \leq \ell$. The regions corresponding to the boundary conditions for which the cohomology is trivial are gray shaded, for the boundary conditions corresponding to the white regions there is a unique non-trivial cohomology element. The labels correspond to different cases we consider in the proposition and in the proof: (1) $c_1+c_N \geq N-2$, (i)  $c_1+c_N < N-2$ and $c_1, c_N \geq p-1$, (i*)  $c_1 \geq p-1$ and $c_N < p-1$ or vice versa, (ii) $c_1,c_N <p-1$ and $c_1+c_N \geq p-2$, (iii) $c_1,c_N <p-1$ and $c_1+c_N < p-2$. \label{fig:specialshort2}}
\end{center}
\end{figure}

\begin{proof}
Let us write $N=\ell+p+2$ with $1 \leq p \leq \ell$. We take $S_1$ to consist of the first and last site, and $S_2$ then consists of the remaining $\ell+p$ sites. This can be illustrated as follows:
\begin{center}
  \begin{tikzpicture}      
     \foreach \x in {3.6,6.3}
     {
     \draw[thick,xshift=\x cm] (0,0) rectangle (.3,.3);
     }
     \foreach \x in {3.6}
     {
     \draw[dotted,xshift=\x cm] (0.3,0.3)--(2.7,.3);
     \draw[dotted,xshift=\x cm] (0.3,0.)--(2.7,0.); 
     \draw[xshift=\x cm] (1.5,-.3) node{$\underset{\ell+p\,\text{sites}}{\underbrace{\hspace{2.3cm}}}$};
     }

  \end{tikzpicture}
\end{center}
The computation of $\coh_{Q_1}$ is very similar to that considered in \cref{app:specialshort1} for the same choice of sublattice (choice (2)). First, using the cohomology of a single site, we find that the cohomology of $Q_1$ is non-trivial provided that a cluster of length $a \geq c_1$ starts on site 2 and a cluster of length $b\geq c_N$ ends on site $N-1$. Second, we again distinguish two cases:
\begin{itemize}
\item[1)] $c_1+c_N \geq N-2$,
\item[2)] $c_1+c_N < N-2$.
\end{itemize}
For case 1) we find that the cohomology of $Q_1$ is non-trivial if and only if all $S_2$ sites are occupied. However, this is not possible since $S_2$ consists of $\ell+p$ consecutive sites and $p>0$. Note that this is different from what we found in \cref{app:specialshort1}, where having all $S_2$ sites occupied \textit{was} allowed. Here, we find that $\coh_{Q_1}$, and thus $\coh_{Q_2} (\coh_{Q_1})$, is trivial for case 1). We conclude that $\coh_Q$ is trivial when $c_1+c_N \geq N-2$. This corresponds to the gray shaded region labeled by (1) in  \cref{fig:specialshort2}.

For case 2) we find that all representatives of $\coh_{Q_1}$ can be written as
\beq
\ket{\psi} = \ket{\bm{0} \underset{c_1}{\underbrace{1\dots1}}} \otimes \ket{\psi '} \otimes \ket{\underset{c_N}{\underbrace{1\dots1}}\bm{0}},
\nonumber
\eeq
where $\ket{\psi'}$ is any state on the middle sites ($c_1+2, c_1+3, \dots, N-1-c_N$) such that in the full representative, $\ket{\psi}$, every particle cluster on $S_2$ has length at most $\ell$. It is not difficult to verify that $\ket{\psi'} \in V_{N'}^{(s')}$ with $N'=N-2-c_1-c_N$ and the special boundary conditions $s' = (c_1'=\min(\ell-c_1,N'-1), c'_{N'} =\min(\ell - c_N,N'-1))$. In particular, using $N'=\ell+p-c_1-c_N$ we find
\beq
c_1' = 
\begin{cases} 
N'-1, & c_N \geq p-1 \\
\ell-c_1, &  c_N < p-1
\end{cases}
 , \quad
c_{N'}' =  
\begin{cases} 
N'-1, & c_1 \geq p-1 \\
\ell-c_N, & c_1 < p-1
\end{cases} .
\nonumber
\eeq

We now turn to solving $\coh_{Q_2} (\coh_{Q_1})$. It is clear that this is equivalent to solving the cohomology of a chain of length $N'$ and special boundary conditions $s'=(c_1',c_{N'}' )$ with $c_1',c_{N'}'$ given above. We solve this problem by considering the following cases:
\begin{itemize}
\item[\textit{(i)}] $c_1 \geq p-1$ and $c_{N} \geq p-1$, 
\item[\textit{(i*)}] $c_1 \geq p-1$ and $c_{N} < p-1$ and vice versa,
\item[\textit{(ii)}] $c_1, c_N < p-1$ and $c_1+c_N \geq p-2$,
\item[\textit{(iii)}] $c_1, c_N < p-1$ and $0 \leq c_1+c_N < p-2$.
\end{itemize}
The labels correspond to the labels in \cref{fig:specialshort2}.

Case i) implies $c_1' = c_{N'}' = N'-1$ and case i*) implies $c_1' = N'-1$ and $c_{N'}' < N'-1$ or vice versa. In both cases we have $N' \leq \ell+1$. The solution to these cohomology problems can be found in \cref{lem:special}. The cohomology has dimension one if $c_1' = c_{N'}' = N'-1$ and is trivial otherwise. The non-trivial element has $N'-1$ particles. 

Case ii) implies $N' \leq \ell+2$, $(c_1',c_{N'}') = (\ell-c_1,\ell-c_N)$ and $c_1', c_{N'}' \leq N'-2$. The solution to these cohomology problems can also be found in \cref{lem:special}.  The cohomology has dimension one if $c_1' + c_{N'}' \geq N'-2$ and is trivial otherwise. The non-trivial element has $N'-2$ particles.

Case iii) implies $\ell+2 < N' \leq \ell+p$, $(c_1',c_{N'}') = (\ell-c_1,\ell-c_N)$ and $c_1', c_{N'}' \leq N'-2$. We do not yet have a general solution for the cohomology problem of open chains of length $\ell+2 < N' \leq \ell+p$. However, note that $c_1' + c_{N'}' = 2 \ell -c_1 -c_N  = N' +\ell -p$ and, consequently, $c_1' + c_{N'}' \geq N'$ using the fact that $p \leq \ell$. This cohomology problem was addressed in the first paragraph of this proof (see case 1) above). We found that the cohomology is trivial in this case.

The solution of cases i) and i*) implies that $\coh_{Q_2} (\coh_{Q_1})$ has dimension one when $c_1, c_N \geq p-1$ and $c_1+c_N < N -2$. The non-trivial element has $N-3$ particles. For $c_1 \geq p-1$ and $c_N < p-1$ (or equivalently $c_1 < p-1$ and $c_N \geq p-1$) the cohomology, $\coh_{Q_2} (\coh_{Q_1})$, is trivial. The solution of case ii) implies that $\coh_{Q_2} (\coh_{Q_1})$ has dimension one for $c_1, c_N < p-1$ and $c_1+c_N \geq p-2$. The non-trivial element has $N-4$ particles. Finally, from case iii) we find that $\coh_{Q_2} (\coh_{Q_1})$ is trivial for $0 \leq c_1+c_N < p-2$.

For all cases we find that $\coh_{Q_2}(\coh_{Q_1})$ is either trivial or has dimension one. It follows that $\coh_Q \simeq \coh_{Q_2}(\coh_{Q_1})$ from \cref{prop:ttt}.
\qedhere
\end{proof}

\subsection{Periodic boundary conditions and $1\leq N\leq \ell+2$}

\begin{repproposition}{lem:periodicshort1}
The cohomology, $\coh_Q (V^{(p)}_N)$, of a chain with periodic boundary conditions has dimension one for $1\leq N\leq \ell+1$ and dimension $\ell+1$ for $N=\ell+2$. The non-trivial elements contain $N-1$ particles.
\end{repproposition}

\begin{proof}
For $1\leq N\leq \ell$ we find that all configurations are allowed except the configuration where all sites are occupied. This is because, due to the closed boundary conditions, the latter configuration corresponds to a cluster of infinite size. Furthermore, for $N=\ell+1$ we also have that all configurations are allowed except the configuration where all sites are occupied. We can now easily solve the cohomology problem for $1\leq N \leq \ell+1$ by taking $S_1$ to be a single site and $S_2$ the remaining $N-1$ sites. The cohomology of $Q_1$ then has dimension one and the non-trivial element has all $S_2$ sites occupied. We then immediately find that $\coh_{Q_2}(\coh_{Q_1})$ also has dimension one and the non-trivial element has $N-1$ particles and $\coh_Q \simeq \coh_{Q_2}(\coh_{Q_1})$. 

For $N=\ell+2$ we also choose $S_1$ to be a single site and $S_2$ the remaining $\ell+1$ sites. The cohomology of $Q_1$ is non-trivial when the $S_1$ site has a cluster of $a$ particles to its left, and a cluster of $b$ particles to its right, such that the sum of their lengths is at least $\ell$. Since $S_2$ can at most be occupied by $\ell$ particles, we find that all the non-trivial elements have $a+b=\ell$. We conclude that there are $\ell+1$ linearly independent non-trivial elements in $\coh_{Q_1}$. All these elements have $\ell$ particles and are thus all also linearly independent non-trivial elements of  $\coh_{Q_2}(\coh_{Q_1})$. Finally, we have $\coh_Q \simeq \coh_{Q_2}(\coh_{Q_1})$ from  \cref{prop:ttt}.
\qedhere
\end{proof}

\subsection{Periodic boundary conditions and $\ell+3\leq N\leq 2\ell+2$}

\begin{repproposition}{lem:periodicshort2}
The cohomology, $\coh_Q (V^{(p)}_N)$, of a chain of length $\ell+3\leq N\leq 2\ell+2$ with periodic boundary conditions has dimension one. The non-trivial elements contain $N-3$ particles.
\end{repproposition}

\begin{proof}
The proof is quite involved and technical, so we first briefly outline the various steps. The proof revolves around the computation of $\coh_{Q_2}(\coh_{Q_1})$, where we take $S_2$ to be $\ell$ consecutive sites and $S_1$ to be the remaining $N_1 \equiv N-\ell$ sites. This can be depicted as follows:
\begin{center}
  \begin{tikzpicture}      
     \draw[thick] (0,0) rectangle (1.8,0.3);
     \draw (0.9,.6) node{$\overset{N_1 = N - \ell \,\text{sites}}{\overbrace{\hspace{1.7cm}}}$};
     \draw (.9,-.3) node{$S_1$};  
     \draw[dotted] (1.8,0.3)--(4.5,.3);
     \draw[dotted] (1.8,0.)--(4.5,0.); 
     \draw[thick] (4.5,0)--(4.5,0.3);
     \draw (3.2,.6) node{$\overset{\ell\,\text{sites}}{\overbrace{\hspace{2.6cm}}}$};
      \draw (3.2,-.3) node{$S_2$};
  \end{tikzpicture}
\end{center}
Here the sites inside the drawn rectangle belong to $S_1$. The right boundary is a drawn line to indicate that it should be identified with the left boundary. Note that $3 \leq N_1 \leq \ell+2$.
The first step is to compute $\coh_{Q_1}$ and obtain representatives of the non-trivial classes. We will find that these representatives fall into two categories, which we call case a) and case b). We then turn to $\coh_{Q_2}(\coh_{Q_1})$. First, we establish that the two cases a) and b) can be considered independently. Second, we consider $\coh_{Q_2}(\coh_{Q_1})$ for case a) and find that the cohomology problem is equivalent to a solved cohomology problem. Third, we consider $\coh_{Q_2}(\coh_{Q_1})$ for case b), which is quite complicated because no such equivalence exists. Finally, we use the tic-tac-toe lemma to obtain $\coh_Q$. 

\paragraph{Cohomology of $Q_1$.} The cohomology problem for $Q_1$ is equivalent to that of an open chain of length $N_1$ with special boundary conditions, $(c_1,c_{N_1})$, that derive from the occupation of $S_2$. Since $N_1 \leq \ell+2$ the solution to this problem can be found in \cref{lem:special}. First, note that if all $\ell$ sites of $S_2$ are occupied we have $c_1=c_{N_1}=0$. From \cref{lem:special} we find that the cohomology is trivial in this case. Second, we consider the case where not all sites of $S_2$ are occupied. The special boundary conditions derive from the length of the cluster that starts on the first site of $S_2$, $k$, and the length of the cluster that ends on the last site of $S_2$, $m$, where
\beq
m+k < \ell, \nonumber
\eeq
to exclude the case where all $S_2$ sites are occupied. We can illustrate this as follows (for convenience we illustrate a case where $m+k \leq \ell-2$, in general, however $m+k < \ell$):
\begin{center}
  \begin{tikzpicture}      
     \draw[thick] (0,0) rectangle (1.8,0.3);
     \draw[] (0.9,-.3) node{$\underset{N_1 \,\text{sites}}{\underbrace{\hspace{1.7cm}}}$};
     \draw[dotted] (1.8,0.3)--(4.5,.3);
     \draw[dotted] (1.8,0.)--(4.5,0.); 
     \draw[thick] (4.5,0)--(4.5,0.3);
     \node at (2.4,0.15) {1 \dots 10};
     \node at (3.9,0.15) {01 \dots 1};
     \draw[] (2.3,-.3) node{$\underset{k \,\text{sites}}{\underbrace{\hspace{0.8cm}}}$};
     \draw[] (4,-.3) node{$\underset{m \,\text{sites}}{\underbrace{\hspace{0.8cm}}}$};
  \end{tikzpicture}
\end{center}
Since $Q_1$ acts only on $S_1$ it if possible to consider the action of $Q_1$ in subspaces with fixed $m,k$. To infer the special boundary conditions for given $m,k$ we first consider the case where $m+k+N_1 \leq \ell$, which implies that the configuration with all $N_1$ sites occupied is allowed. Therefore in this case the cohomology problem for $Q_1$ corresponds to the unconstrained case of \cref{lem:special} and we find that the cohomology is trivial. From now on we thus restrict to $m+k+N_1 > \ell$, which together with $m+k < \ell$ gives
\beq
\ell-N_1 < m+k < \ell. \nonumber
\eeq
In this case, the configuration with all $N_1$ sites occupied is not allowed and we find $c_1 = \min(\ell-m, N_1-1)$ and $c_{N_1} = \min(\ell-k, N_1-1)$, which implies
\beq
c_1 = 
\begin{cases} 
N_1-1, & m \leq \ell+1-N_1 \\
\ell-m, &  m > \ell+1-N_1
\end{cases}
 , \quad
c_{N_1} =  
\begin{cases} 
N_1-1, & k \leq \ell+1-N_1 \\
\ell-k, &  k > \ell+1-N_1
\end{cases} .
\nonumber
\eeq
Using \cref{lem:special} we find that the cohomology of $Q_1$ has dimension one when
\begin{itemize}
\item[a)]{$m,k > \ell+1- N_1$,}
\item[b)]{$m,k \leq \ell+1- N_1$,}
\end{itemize}
with $\ell-N_1< m+k < \ell$. The cohomology of $Q_1$ is trivial otherwise. The non-trivial equivalence classes can be represented by
\beq
\ket{\psi_{k,m}} =\ket{X} \otimes  \begin{cases} \ket{\underset{k \,\text{sites}}{\underbrace{1 \dots 1}}  0 \underset{m \,\text{sites}}{\underbrace{1 \dots 1}}} & m+k = \ell-1 \\
\ket{\underset{k \,\text{sites}}{\underbrace{1 \dots 1}}  00 \underset{m \,\text{sites}}{\underbrace{1 \dots 1}}} & m+k = \ell-2 \\
\ket{\underset{k \,\text{sites}}{\underbrace{1 \dots 1}}0} \otimes \ket{\psi'} \otimes \ket{0 \underset{m \,\text{sites}}{\underbrace{1 \dots 1}}} & \ell-N_1<  m+k < \ell-2 
\end{cases}
\nonumber
\eeq
where $\ket{X}$ has support on all $N_1$ sites of sublattice $S_1$ and for case a) $\ket{X} = \ket{A}\equiv \ket{01 \dots 10}$ and for case b) $\ket{X} = \ket{B}\equiv \ket{01 \dots 1}$. Finally, the state $\ket{\psi'}$ has support on the sites $N_1+k+2, \dots, N_1+\ell-m-1$, which belong to $S_2$. Note that the $\ket{\psi_{k,m}}$ are representatives for any choice of the state, $\ket{\psi'}$ and that the $\ket{\psi'}$ belong to the unconstrained Hilbert space spanned by all possible configuration on $\ell-k-m-2$ sites. In the following we will assume without loss of generality that $\ket{\psi'}$ has a definite fermion number.

The structure of $\coh_{Q_1}$ is illustrated in \cref{fig:closedQ1}.
\begin{figure}[ht]
\begin{center}
\includegraphics[width = 0.45\textwidth]{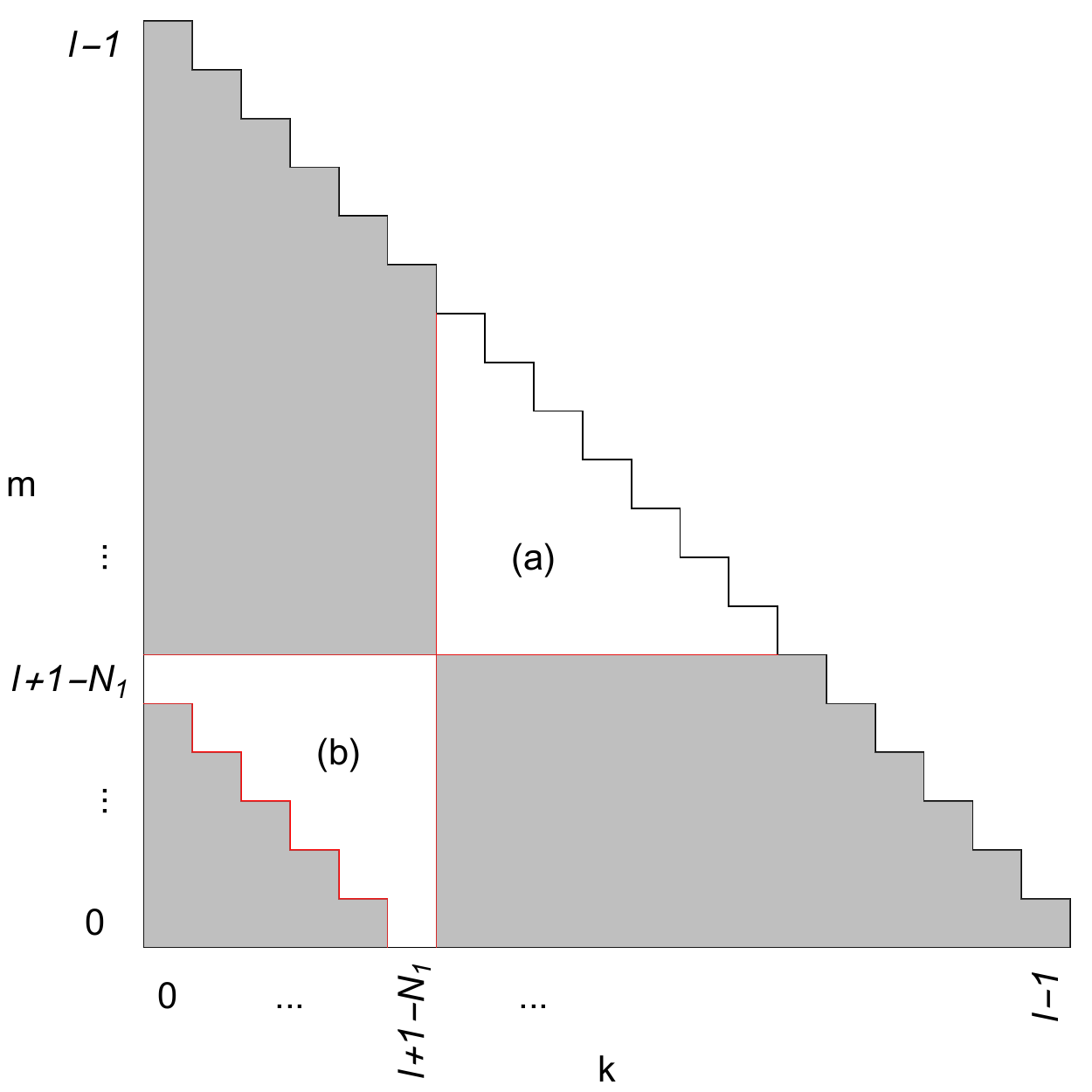}
\includegraphics[width = 0.45\textwidth]{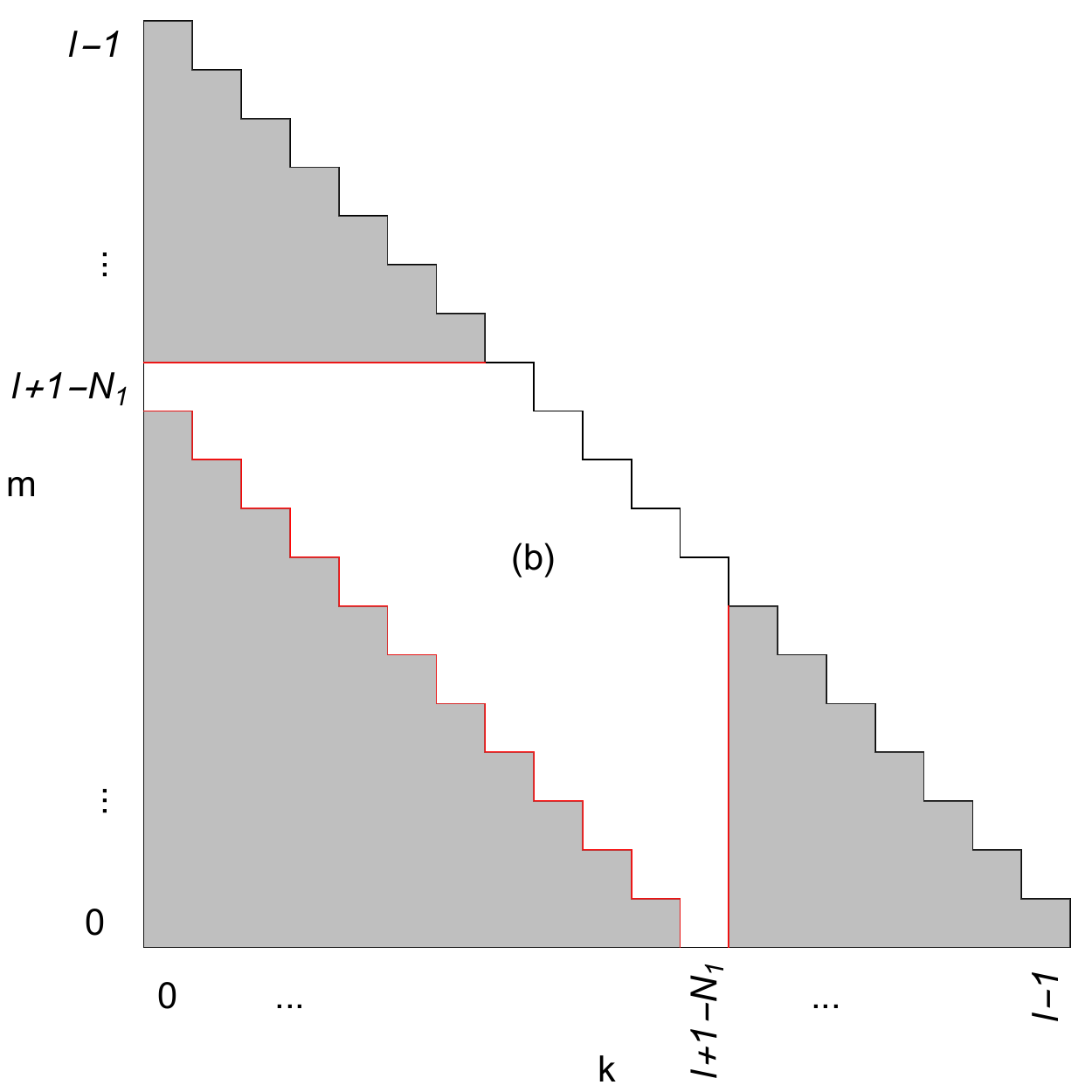}
\caption{We illustrate the structure of the cohomology of $Q_1$ for closed chains of length $N=N_1+\ell$ and $3 \leq N_1 \leq \ell+2$ with periodic boundary conditions for different values of $k,m$. Here  the sublattice $S_1$ consists of $N_1$ consecutive sites, $k$ is the length of the cluster that starts on the first site of $S_2$ and $m$ the length of the cluster that ends on the last site of $S_2$. Note that $m+k<\ell$. The regions for which $\coh_{Q_1}$ is trivial are gray shaded, for the values of $k,m$ corresponding to the white regions $\coh_{Q_1}$ has at least one non-trivial element. The diagram on the left corresponds to $\ell/2+2<N_1 \leq \ell+2$, the diagram on the right is for $3 \leq N_1 \leq \ell/2+2$. The labels correspond to different cases we consider in the proof: (a) $m,k > \ell+1- N_1$, (b) $m,k \leq \ell+1- N_1$ and $m+k>\ell-N_1$. \label{fig:closedQ1}}
\end{center}
\end{figure}

\paragraph{$\coh_{Q_2}(\coh_{Q_1})$: independence of cases a) and b).} We now turn to computing $\coh_{Q_2}(\coh_{Q_1})$. We will first show that the cases a) and b) can be considered separately. First, we note that for any representative, $\ket{\psi_{k,m}}$, of case b) and any representative, $\ket{\psi_{k',m'}}$, of case a), we have $k'>k$ and $m'>m$. This can be seen in \cref{fig:closedQ1}, where it is clear that to go from region b) to region a) you have to move both up and to the right. Second, we show that up to coboundaries of $Q_1$, $Q_2$ maps a representative $\ket{\psi_{k,m}}$ onto a superposition of representatives $\ket{\psi_{k',m'}}$ with either $k'=k$ and $m' \geq m$ or $k' \geq k$ and $m' = m$. To see this we consider the action of $Q_2$ on a representative, $\ket{\psi_{k,m}}$, of $\coh_{Q_1}$ for case b). It is easy to see that for $m+k = \ell-1$, we find $Q_2 \ket{\psi_{k,m}} = 0$. For $m+k = \ell-2$, we find 
\beq
Q_2 \ket{\psi_{k,\ell-2-k}} = (-1)^{N_1-1+k} ( (1-\delta_{k+N_1,\ell+1}) \ket{\psi_{k+1,\ell-2-k}}+\ket{\psi_{k,\ell-1-k}}),
\nonumber
\eeq
where the Kronecker delta ensures that the first term is absent when $k+N_1=\ell+1$, this is necessary because $\ket{B}$ ends with a cluster of length $N_1-1$. Finally, for $\ell-N_1<  m+k < \ell-2$ we find 
\bea
Q_2 \ket{\psi_{k,m}} &=& (-1)^{N_1-1+k} \ket{B} \otimes \left(
 (1-\delta_{k+N_1,\ell+1}) \ket{\underset{k+1 \,\text{sites}}{\underbrace{1 \dots 1}}} \otimes \ket{\psi'} \otimes \ket{0 \underset{m \,\text{sites}}{\underbrace{1 \dots 1}}} \right. \nonumber\\
& &+  \ket{\underset{k \,\text{sites}}{\underbrace{1 \dots 1}}0} \otimes Q_2 \ket{\psi'} \otimes \ket{0 \underset{m \,\text{sites}}{\underbrace{1 \dots 1}}}
 + \left. (-1)^{x} \ket{\underset{k \,\text{sites}}{\underbrace{1 \dots 1}}0} \otimes \ket{\psi'} \otimes \ket{\underset{m+1 \,\text{sites}}{\underbrace{1 \dots 1}}} \right) \nonumber\\
&=& (-1)^{N_1-1+k} \Big(  \sum_{k<k' \leq \ell+1-N_1} c_{k'} \ket{\tilde{\psi}_{k',m}}\nonumber\\
& & + \sum_{m<m' \leq \ell+1-N_1} c_{m'} \ket{\tilde{\psi}_{k,m'}} +  \ket{\tilde{\psi}_{k,m}} \Big)  + Q_1 \ket{\phi}.
\nonumber
\eea
We introduced $x$ as the number of particles in $\ket{\psi'}$ and $c_{k'}$ are constants depending on $\ket{\psi'}$. The tildes in the last line indicate that these representatives have a different state, $\ket{\tilde{\psi}'}$, on the sites $N_1+k+2, \dots, N_1+\ell-m-1$. Finally, we imposed an upper bound on $k'$ and $m'$ in the sums by using the fact that a state with $k\leq \ell+1-N_1<m$ or $m\leq \ell+1-N_1<k$ is a coboundary of $Q_1$. This leads to the term $Q_1 \ket{\phi}$, where $\ket{\phi}$ is some state. We find that all representatives $\ket{\tilde{\psi}}$ belong to case b). We conclude that $Q_2$ cannot map a state that belongs to case b) to a state that belongs to case a) or vice versa. We can thus treat these two cases independently. 

\paragraph{$\coh_{Q_2}(\coh_{Q_1})$: case a).}
For case a) we find that $Q_2$ acts trivially on the first $\ell+2- N_1$ and last $\ell+2- N_1$ sites of $S_2$ since $m,k \geq \ell+2- N_1$. Furthermore, since $m+k<\ell$ we find that any configuration on the remaining sites is allowed, except all occupied. Combining the lower bound on $k,m$ and the strict upper bound on their sum, we find that case a) only occurs when $\ell>k+m\geq 2(\ell+2-N_1)$, which implies $N>3\ell/2+2$. For these cases, we find that $Q_2$ effectively acts on a chain of length $N_{\textrm{eff}} \equiv \ell-2(\ell+2- N_1) \leq \ell$ for which all but the completely full configuration is allowed since $m+k < \ell$. The solution to this problem can be found in \cref{lem:special} using special boundary conditions $s=(N_{\textrm{eff}}-1,N_{\textrm{eff}}-1)$. We find that the dimension of the cohomology is one. We conclude that case a) contributes one non-trivial element to $\coh_{Q_2}(\coh_{Q_1})$ when $N>3\ell/2+2$ and none otherwise. There exists a representative of this class that has $N-3$ particles.

\paragraph{$\coh_{Q_2}(\coh_{Q_1})$: case b).}
For case b) the computation is more complicated. We first introduce a suitable choice for the states $\ket{\psi'}$. Let us denote $Q_2$ restricted to the sites $N_1+k+2, \dots, N_1+\ell-m-1$ by $\bar{Q}_2$. Remembering that the states $\ket{\psi'}$ belong to the unconstrained Hilbert space, we easily check that $\coh_{\bar{Q}_2}$ vanishes for given $k,m$. It follows that there is a basis of doublet representations of $\bar{Q}_2$ and we can choose to write any state $\ket{\psi'}$ as a state $\ket{\chi_i}$ such that
\beq
\bar{Q}_2 \ket{\chi_{2r-1}} = \ket{\chi_{2r}} \quad \textrm{and} \quad \bar{Q}_2 \ket{\chi_{2r}} = 0,
\nonumber
\eeq
for $r=1, \dots, 2^{\ell-m-k-3}$. Finally, let us write the representative $\ket{\psi_{m,k}}$ with $\ket{\psi'}=\ket{\chi_i}$ as $\ket{\psi_{m,k,i}}$. Considering again the action of $Q_2$ on such a representative with $i=2r-1$, we find
\bea
Q_2 \ket{\psi_{k,m,2r-1}} &=& (-1)^{N_1-1+k} \Big(  \sum_{k',m',j} c_{k',m',j} \ket{\psi_{k',m',j}} +  \ket{\psi_{k,m,2r}} \Big)  + Q_1 \ket{\phi}, \nonumber
\eea
for some state $\ket{\phi}$, some constants $c_{k',m',j}$ and $k',m' \leq \ell+1-N_1$ and $k'+m'>k+m$. Note that the first term is absent for $k=m=\ell+1-N_1$. We can now conclude the following for $\ell-N_1 < m+k <\ell-2$:
\begin{itemize}
\item[i)] $\ket{\psi_{k,m,2r-1}}$ is not in the kernel of $Q_2$,
\item[ii)] $\ket{\psi_{k,m,2r}}$ is a representative of the trivial class of $\coh_{Q_2}(\coh_{Q_1})$ for $m = k = \ell+1-N_1$,
\item[iii)] $\ket{\psi_{k,m,2r}}$ is in the same equivalence class as a superposition of representatives $\ket{\psi_{k',m',i}}$ with $k'+m'>k+m$ unless $m = k = \ell+1-N_1$.
\end{itemize}

We will now use these properties to compute $\coh_{Q_2}(\coh_{Q_1})$ for two separate cases: 
\begin{itemize}
\item[b1)] $N> 3\ell/2+2$,
\item[b2)] $N \leq 3\ell/2+2$.
\end{itemize}
An example of case b1) (case b2)) is given on the left (right) of \cref{fig:closedQ1}. Note that for case b2) case a) is absent.

For case b1) we have that $N_1> \ell/2+2$ and therefore $k+m<\ell-2$ since $k,m \leq \ell+1-N_1$. We now prove by induction that all $\ket{\psi_{k,m,i}}$ are representatives of the trivial class of $\coh_{Q_2}(\coh_{Q_1})$. First, it follows from i) and ii) that all $\ket{\psi_{k,m,i}}$ with $m = k = \ell+1-N_1$ are representatives of the trivial class of $\coh_{Q_2}(\coh_{Q_1})$. In \cref{fig:closedQ1} $m = k = \ell+1-N_1$ is the upper-right corner of the region b). We now assume that for $m + k > s$, with $s$ some integer such that $\ell-N_1 < s < 2(\ell+1-N_1)$, all $\ket{\psi_{k,m,i}}$ are representatives of the trivial class of $\coh_{Q_2}(\coh_{Q_1})$. In \cref{fig:closedQ1} $m+k=s$ is a diagonal and we assume that the part of region b) above and to the right of this diagonal contains only representatives of the trivial class. Using i) and iii) we can then prove that all $\ket{\psi_{k,m,i}}$ with $m +k=s$ are also representatives of the trivial class. From i) we know that the states with $k+m=s$ and $i$ odd are not in the kernel of $Q_2$ and thus are representatives of the trivial class. From iii) we know that the states with $k+m=s$ and $i$ even are in the same equivalence class as states with $k+m>s$, which by assumption are representatives of the trivial class. Since our assumption is true for $s = 2(\ell+1-N_1)$, we conclude that all $\ket{\psi_{k,m,i}}$ with $\ell-N_1 < m+k < 2(\ell+1-N_1)$ are representatives of the trivial class of $\coh_{Q_2} (\coh_{Q_1})$.

For case b2) we prove by induction that all $\ket{\psi_{k,m,i}}$ are representatives of either the trivial class or of a unique non-trivial class of $\coh_{Q_2}(\coh_{Q_1})$. First, we show that this holds for the representatives with $k+m = \ell-2,\ell-1$ (note that these were absent in case b1)). Remembering the action of $Q_2$ on these representatives we conclude that:
\begin{itemize}
\item[1)] $\ket{\psi_{k,\ell-k-1}} \in \im Q_2$ for $k=\ell+1-N_1$ and
\item[2)] $\ket{\psi_{k,\ell-k-1}}$ is in the same equivalence class as $\ket{\psi_{k-1,\ell-k}}$ for $N_1-2 \leq k < \ell+1-N_1$.
\end{itemize}
It follows that the $\ket{\psi_{k,\ell-k-1}}$ with $N_1-2\leq k < \ell+1-N_1$ are representatives of the same equivalence class as $\ket{\psi_{k,\ell-k-1}}$ with $k=\ell+1-N_1$, which is the trivial class. Furthermore, it follows that there is only one state with $m+k=\ell-2$  that is in the kernel of $Q_2$:
\beq
Q_2  \ket{\Psi} =0, \quad \ket{\Psi} \equiv \sum_{k=N_1-3}^{\ell+1-N_1} \ket{\psi_{k,\ell-k-2}}. \nonumber
\eeq

Note that for $2(\ell+1-N_1) = \ell-2$, that is $N = 3\ell/2+2$, which requires $\ell$ even, there are no states with $k+m = \ell-1$ and there is precisely one state with $k+m=\ell-2$:  $\ket{\Psi} = \ket{\psi_{\ell+1-N_1,\ell+1-N_1}}$ which is in the kernel of $Q_2$ within $\coh_{Q_1}$. 

The inductive step is very similar to that in case b1). We assume that for $m + k > s$, with $s$ some integer such that $\ell-N_1 < s < \ell-1$, all $\ket{\psi_{k,m,i}}$ are either representatives of the trivial class, $[0]$, or of the same equivalence class as $\ket{\Psi}$, $[\ket{\Psi}]$. From i) we know that the states with $k+m=s$ and $i$ odd are not in the kernel of $Q_2$ and thus are representatives of the trivial class. From iii) we know that the states with $k+m=s$ and $i$ even are in the same equivalence class as states with $k+m>s$, which by assumption are representatives of either $[0]$ or $[\ket{\Psi}]$. Since our assumption is true for $s = \ell-2$ and the states with $m+k=\ell-1$ are representatives of $[0]$, we conclude that all $\ket{\psi_{k,m,i}}$ with $k+m > \ell-N_1$ and $k,m \leq \ell+1-N_1$ are representatives of either $[0]$ or $[\ket{\Psi}]$. We thus conclude that for $N \leq 3\ell/2+2$ the dimension of $\coh_{Q_2} (\coh_{Q_1})$ is one and the non-trivial class has a representative with $\ell-2+N_1-1 = N-3$ particles.

Combining cases a) and b) we find that $\coh_{Q_2}(\coh_{Q_1})$ has dimension one, which gives $\coh_Q \simeq \coh_{Q_2}(\coh_{Q_1})$ from \cref{prop:ttt}.

\qedhere
\end{proof}

\end{document}